\documentclass{llncs} 

\usepackage{mathpartir}
\usepackage{amsmath,amssymb,mathtools,xspace,upgreek,url}
\usepackage{float,bm,color}

\floatstyle{boxed} 
\restylefloat{figure}

\usepackage{envBisim}

\title{Environmental Bisimulations\\ for Delimited-Control Operators}
\author{Dariusz Biernacki \inst{1} \and Sergue\"i Lenglet \inst{2}}

\institute{Institute of Computer Science, University of Wroc\l{}aw \and LORIA,
  Universit\'e de Lorraine}

\begin{document}

\maketitle

\begin{abstract} 
  We present a theory of environmental bisimilarity for the
  delimited-control operators {\it shift} and {\it reset}. We consider
  two different notions of contextual equivalence: one that does not
  require the presence of a top-level control delimiter when executing
  tested terms, and another one, fully compatible with the original
  CPS semantics of shift and reset, that does. For each of them, we
  develop sound and complete environmental bisimilarities, and we
  discuss up-to techniques.
\end{abstract}

\section{Introduction}


Control operators for delimited
continuations~\cite{Danvy-Filinski:LFP90,Felleisen:POPL88} provide
elegant means for expressing advanced control
mechanisms~\cite{Danvy-Filinski:LFP90,Hieb-al:LaSC93}. Moreover, they
play a fundamental role in the semantics of computational
effects~\cite{Filinski:POPL94}, normalization by
evaluation~\cite{Balat-al:POPL04} and as a crucial refinement of
abortive control operators such as {\it
  callcc}~\cite{Felleisen:POPL88,Sitaram-Felleisen:LFP90}. Of special
interest are the control operators {\it shift} and {\it
  reset}~\cite{Danvy-Filinski:LFP90} due to their origins in
continuation-passing style (CPS) and their connection with
computational monads -- as demonstrated by
Filinski~\cite{Filinski:POPL94}, shift and reset can express in direct
style arbitrary computational effects, such as mutable state,
exceptions, etc. Operationally, the control delimiter reset delimits
the current continuation and the control operator shift abstracts the
current delimited continuation as a first class value that when
resumed is composed with the then-current continuation.

Because of the complex nature of control effects, it can be difficult
to determine if two programs that use shift and reset are equivalent
(i.e., behave in the same way) or not. \emph{Contextual
  equivalence}~\cite{JHMorris:PhD} is widely considered as the most
natural equivalence on terms in languages similar to the
$\lambda$-calculus. Roughly, two terms are contextually equivalent if
we cannot tell them apart when they are executed within any
context. The latter quantification over contexts makes this relation
hard to use in practice, so we usually look for simpler
characterizations of contextual equivalence, such as coinductively
defined \emph{bisimilarities}.

In our previous work, we defined
\emph{applicative}~\cite{Biernacki-Lenglet:FOSSACS12} and \emph{normal
  form}~\cite{Biernacki-Lenglet:FLOPS12} bisimilarities for shift and
reset. Applicative bisimilarity characterizes contextual equivalence,
but still quantifies over some contexts to relate terms (e.g.,
$\lambda$-abstractions are applied to the same arbitrary argument). As
a result, some equivalences remain quite difficult to prove. In
contrast, normal form bisimilarity does not contain any quantification
over contexts or arguments in its definition: the tested terms are
reduced to normal forms, which are then decomposed in bisimilar
sub-terms. Consequently, proofs of equivalence are usually simpler
than with applicative bisimilarity, and they can be simplified even
further with \emph{up-to techniques}. However, normal form
bisimilarity is not \emph{complete}, i.e., there exists contextually
equivalent terms which are not normal form bisimilar.

\emph{Environmental bisimilarity}~\cite{Sangiorgi-al:TOPLAS11} is a
different kind of behavioral equivalence which in terms of strength
and practicality
can be situated
in between applicative and normal form
bisimilarities. It has originally been proposed
in~\cite{Sumii-Pierce:TCS07} and has been since defined in various
higher-order languages (see, e.g.,
\cite{Sato-Sumii:APLAS09,Sumii:TCS10,Pierard-Sumii:LICS12}). Like
applicative bisimilarity, it uses some particular contexts to test
terms, except that the testing contexts are built from an environment,
which represents the knowledge built so far by an outside
observer. Environmental bisimilarity usually characterizes contextual
equivalence, but is harder to establish than applicative
bisimilarity. Nonetheless, like with normal form bisimilarity, one can
define powerful up-to techniques~\cite{Sangiorgi-al:TOPLAS11} to
simplify the equivalence proofs.
Besides, the authors of~\cite{Koutavas-al:ENTCS11} argue that the
additional complexity of environmental bisimilarity is necessary to
handle more realistic features, like local state or exceptions.

In the quest for a powerful enough (i.e., as discriminative as
contextual equivalence) yet easy-to-use equivalence for delimited
control, we study in this paper the environmental theory of a calculus
with shift and reset. More precisely, we consider two semantics for
shift and reset: the original one~\cite{Biernacka-al:LMCS05}, where
terms are executed within a top-level reset, and a more relaxed
semantics where this requirement is lifted. The latter is commonly
used in implementations of shift and
reset~\cite{Dybvig-al:JFP06,Filinski:POPL94} as well as in some
studies of these
operators~\cite{Asai-Kameyama:APLAS07,Kameyama:HOSC07}, including our
previous
work~\cite{Biernacki-Lenglet:FOSSACS12,Biernacki-Lenglet:FLOPS12}. So
far, the behavioral theory of shift and reset with the original
semantics has not been studied. Firstly, we define environmental
bisimilarity for the relaxed semantics and study its properties;
especially we discuss the problems raised by delimited control for the
definition of bisimulation up to context, one of the most powerful
up-to techniques. Secondly, we propose the first behavioral theory for
the original semantics, and we pinpoint the differences between the
equivalences of the two semantics. In particular, we show that the
environmental bisimilarity for the original semantics is complete
w.r.t. the axiomatization of shift and reset
of~\cite{Kameyama-Hasegawa:ICFP03}, which is not the case for the
relaxed semantics, as already proved
in~\cite{Biernacki-Lenglet:FOSSACS12} for applicative bisimilarity.

In summary, we make the following contributions in this paper.
\begin{itemize}
\item We show that environmental bisimilarity can be defined for a
  calculus with delimited control, for which we consider two different
  semantics. In each case, the defined bisimilarity equals contextual
  equivalence.
\item For the relaxed semantics, we explain how to handle \emph{stuck
  terms}, i.e., terms where a capture cannot go through because of the
  lack of an outermost reset.
\item We discuss the limits of the usual up-to techniques in the case
  of delimited control.
\item For the original semantics, we define a contextual equivalence,
  and a corresponding environmental bisimilarity. Proving soundness of
  the bisimilarity w.r.t. contextual equivalence requires significant
  changes from the usual soundness proof scheme. We discuss how
  environmental bisimilarity is easier to adapt than applicative
  bisimilarity.
\item We give examples illustrating the differences between the two
  semantics.
\end{itemize}
The rest of the paper is organized as follows: in Section
\ref{s:calculus}, we present the calculus $\lamshift$ used in this
paper, and recall some results, including the axiomatization
of~\cite{Kameyama-Hasegawa:ICFP03}. We develop an environmental theory
for the relaxed semantics in Section~\ref{s:general}, and for the
original semantics in Section~\ref{s:programs}. We conclude in
Section~\ref{s:conclusion}, and the appendices contain the
characterization proofs omitted from the main text.

\section{The Calculus $\lamshift$}
\label{s:calculus}

\subsection{Syntax}
The language $\lamshift$ extends the call-by-value $\lambda$-calculus
with the delimited-control operators \emph{shift} and
\emph{reset}~\cite{Danvy-Filinski:LFP90}. We assume we have a set of
term variables, ranged over by $\varx$, $y$, $z$, and $\vark$. We use
$k$ for term variables representing a continuation (\eg when bound
with a shift), while $x$, $y$, and $z$ stand for any values; we
believe such distinction helps to understand examples and reduction
rules. The syntax of terms is given by the following grammar:
\[
\begin{array}{llll}
  \textrm{Terms:} & \tm & ::= & \varx \Mid \lam \varx \tm \Mid \app
  \tm \tm \Mid \shift \vark \tm \Mid \reset \tm
\end{array}
\]
\noindent \emph{Values}, ranged over by $\val$, are terms of the form
$\lam \varx \tm$. The operator shift ($\shift \vark \tm$) is a capture
operator, the extent of which is determined by the delimiter reset
($\rawreset$). A $\lambda$-abstraction $\lam \varx \tm$ binds $\varx$
in $\tm$ and a shift construct $\shift \vark \tm$ binds $\vark$
in~$\tm$; terms are equated up to $\alpha$-conversion of their bound
variables. The set of free variables of $\tm$ is written $\fv \tm$; a
term $\tm$ is \emph{closed} if $\fv\tm = \emptyset$.

We distinguish several kinds of contexts, represented outside-in, as
follows:
\[
\begin{array}{llll}
  \textrm{Pure contexts:}& \ctx & ::= & \mtctx \Mid \vctx \val \ctx \Mid \apctx \ctx
  \tm \\[2mm]
  \textrm{Evaluation contexts:\:}& \rctx & ::= & \mtctx \Mid \vctx \val \rctx \Mid \apctx \rctx
  \tm \Mid \resetctx \rctx \\[2mm]
  \textrm{Contexts:}& \cctx & ::= & \mtctx \Mid \lam \varx \cctx
  \Mid \vctx \tm \cctx \Mid \apctx \cctx \tm \Mid \shift \vark \cctx \Mid \resetctx
  \cctx 
\end{array}
\]
\noindent Regular contexts are ranged over by $\cctx$. The pure
evaluation contexts\footnote{This terminology comes from Kameyama (\eg
  in \cite{Kameyama-Hasegawa:ICFP03}).} (abbreviated as pure
contexts), ranged over by $\ctx$, represent delimited continuations
and can be captured by shift. The call-by-value evaluation contexts,
ranged over by $\rctx$, represent arbitrary continuations and encode
the chosen reduction strategy. Filling a context $\cctx$ (respectively
$\ctx$, $\rctx$) with a term $\tm$ produces a term, written $\inctx
\cctx \tm$ (respectively $\inctx \ctx \tm$, $\inctx \rctx \tm$); the
free variables of $\tm$ may be captured in the process. We extend the
notion of free variables to contexts (with $\fv \mtctx=\emptyset$),
and we say a context $\cctx$ (respectively $\ctx$, $\rctx$) is
\emph{closed} if $\fv \cctx = \emptyset$ (respectively $\fv \ctx =
\emptyset$, $\fv \rctx = \emptyset$).

\subsection{Reduction Semantics}
\label{ss:reduction}

The call-by-value reduction semantics of $\lamshift$ is defined as
follows, where $\subst \tm \varx \val$ is the usual capture-avoiding
substitution of $\val$ for $\varx$ in $\tm$:
\[
\begin{array}{lrll}
  \RRbeta & \quad \inctx \rctx {\app {\lamp \varx \tm} \val} & \redcbv & \inctx
  \rctx {\subst \tm \varx \val} \\[2mm]
  \RRshift & \quad \inctx \rctx {\reset{\inctx \ctx {\shift \vark \tm}}} &
  \redcbv & \inctx \rctx{\reset{\subst \tm \vark
    {\lam \varx {\reset {\inctx \ctx \varx}}}}} \mbox{ with } \varx \notin \fv
  \ctx\\[2mm]
  \RRreset & \quad \inctx \rctx {\reset \val} & \redcbv & \inctx \rctx \val
\end{array}
\]
\noindent The term $\app {\lamp \varx \tm} \val$ is the usual
call-by-value redex for $\beta$-reduction (rule $\RRbeta$). The
operator $\shift \vark \tm$ captures its surrounding context $\ctx$ up
to the dynamically nearest enclosing reset, and substitutes $\lam
\varx {\reset {\inctx \ctx \varx}}$ for $\vark$ in $\tm$ (rule
$\RRshift$). If a reset is enclosing a value, then it has no purpose
as a delimiter for a potential capture, and it can be safely removed
(rule $\RRreset$). All these reductions may occur within a metalevel
context $\rctx$, so the reduction rules specify both the notion of
reduction and the chosen call-by-value evaluation strategy that is
encoded in the grammar of the evaluation contexts. Furthermore, the
reduction relation $\redcbv$ is compatible with evaluation contexts
$\rctx$, i.e., $\inctx \rctx \tm \redcbv \inctx \rctx {\tm'}$ whenever
$\tm \redcbv {\tm'}$.

There exist terms which are not values and which cannot be reduced any
further; these are called \emph{stuck terms}.
\begin{definition}
  A term $\tm$ is stuck if $\tm$ is not a value and $\tm \not
  \redcbv$.
\end{definition}
For example, the term $\inctx \ctx {\shift \vark \tm}$ is stuck
because there is no enclosing reset; the capture of $\ctx$ by the
shift operator cannot be triggered.

\begin{lemma}
  \label{l:stuck}
  A closed term $\tm$ is stuck iff $\tm = \inctx \ctx {\shift \vark
    {\tm'}}$ for some $\ctx$, $k$, and $\tm'$.
\end{lemma}

\begin{definition}
  A term $\tm$ is a normal form if $\tm$ is a value or a stuck term.
\end{definition}

We call \emph{redexes} (ranged over by $\redex$) terms of the form
$\app{\lamp \varx \tm} \val$, $\reset {\inctx \ctx {\shift \vark
    \tm}}$, and $\reset \val$. Thanks to the following
unique-decomposition property, the reduction relation $\redcbv$ is
deterministic.
\begin{lemma}
  \label{l:unique-decomp}
  For all closed terms $\tm$, either $\tm$ is a normal form, or there
  exist a unique redex $\redex$ and a unique context $\rctx$ such that
  $\tm = \inctx \rctx \redex$.
\end{lemma}
 
Finally, we write $\clocbv$ for the transitive and reflexive closure
of $\redcbv$, and we define the evaluation relation of $\lamshift$ as
follows.
\begin{definition}
  We write $\tm \evalcbv \tm'$ if $\tm \clocbv \tm'$ and $\tm'
  \not\redcbv$.
\end{definition}
The result of the evaluation of a closed term, if it exists, is a normal
form. If a term $\tm$ admits an infinite reduction sequence, we say it
\emph{diverges}, written $\tm \divcbv$. Henceforth, we use $\Omega = \app {\lamp
  \varx {\app \varx \varx}}{\lamp \varx {\app \varx \varx}}$ as an example of
such a term.

\subsection{CPS Equivalence}
\label{s:cps-equiv}

In~\cite{Kameyama-Hasegawa:ICFP03}, the authors propose an equational
theory of shift and reset based on
CPS~\cite{Danvy-Filinski:LFP90}. The idea is to relate terms that have
$\beta\eta$-convertible CPS translations.

\begin{definition}
  Terms $\tmzero$ and $\tmone$ are CPS equivalent, written $\tmzero
  \cpsequiv \tmone$, if their CPS translations are
  $\beta\eta$-convertible.
\end{definition}

Kameyama and Hasegawa propose eight axioms in~\cite{Kameyama-Hasegawa:ICFP03} to
characterize CPS equivalence: two terms are CPS equivalent iff one can derive
their equality using the equations below. Note that the axioms are defined on
open terms, and suppose variables as values.
\[
\begin{array}{rcllrcll}
  \app{\lamp \varx \tm} \val &\eqKH& \subst \tm \varx \val & 
  & \app {\lamp \varx {\inctx \ctx \varx}} \tm &\eqKH& \inctx \ctx \tm \mbox{ if } x
  \notin \fv \ctx &  \\[2mm]
  \reset {\inctx \ctx {\shift \vark \tm}} &\eqKH& \reset {\subst \tm \vark {\lam
      \varx {\reset {\inctx \ctx \varx}}}} & 
  & \quad  \reset {\app {\lamp \varx \tmzero}{\reset \tmone}} &\eqKH& \app {\lamp \varx
    {\reset \tmzero}}{\reset \tmone} &   \\[2mm]
  \reset \val &\eqKH& \val & 
  & \shift \vark {\reset \tm} &\eqKH& \shift \vark \tm &  \\[2mm]
  \lam \varx {\app \val \varx} &\eqKH& \val \mbox{ if } \varx \notin \fv \val &
    & \shift \vark {\app \vark \tm} &\eqKH& \tm \mbox{ if } \vark \notin \fv
  \tm &  
\end{array}
\]

\noindent We use the above relations as examples throughout the
paper. Of particular interest is the axiom $\app {\lamp \varx {\inctx
    \ctx \varx}} \tm \eqKH \inctx \ctx \tm$ (if $x \notin \fv \ctx$),
called $\beta_\Omega$ in~\cite{Kameyama-Hasegawa:ICFP03}, which can be
difficult to prove with
bisimilarities~\cite{Biernacki-Lenglet:FOSSACS12}.

\subsection{Context Closures}

Given a relation $\rel$ on terms, we define two context closures that
generate respectively terms and evaluation contexts. The term
generating closure $\cloct \rel$ is defined inductively as the
smallest relation satisfying the following rules:
\begin{mathpar}
  \inferrule{\tm \rel \tm'}
            {\tm \cloct \rel \tm'}
  \and
  \hspace{-1em}\inferrule{}
            {\varx \cloct \rel \varx}
  \and
  \hspace{-1em}\inferrule{\tm \cloct \rel \tm'}
            {\lam \varx \tm \cloct \rel \lam \varx {\tm'}}
  \and
  \hspace{-1em}\inferrule{\tmzero \cloct\rel \tmzero' \quad \tmone \cloct \rel \tmone'}
            {\app \tmzero \tmone \cloct\rel \app{\tmzero'}{\tmone'}}
  \and
  \hspace{-1em}\inferrule{\tm \cloct \rel \tm'}
            {\shift \vark \tm \cloct \rel \shift \vark {\tm'}}
  \and
  \hspace{-1em}\inferrule{\tm \cloct \rel \tm'}
            {\reset \tm \cloct \rel \reset{\tm'}}
\end{mathpar}

Even if $\rel$ is defined only on closed terms, $\cloct\rel$ is
defined on open terms. In this paper, we consider the restriction of
$\cloct\rel$ to closed terms unless stated otherwise. The context
generating closure $\clocc \rel$ of a relation $\rel$ is defined
inductively as the smallest relation satisfying the following rules:
\begin{mathpar}
  \inferrule{ }
  {\mtctx \clocc\rel \mtctx}
  \and
  \inferrule{\rctxzero \clocc \rel \rctxone \\ \valzero \cloctc\rel \valone}
  {\vctx \valzero \rctxzero \clocc\rel \vctx \valone \rctxone}
  \and
  \inferrule{\rctxzero \clocc \rel \rctxone \\ \tmzero \cloctc\rel \tmone}
  {\apctx \rctxzero \tmzero \clocc\rel \apctx \rctxone \tmone}
  \and
  \inferrule{\rctxzero \clocc\rel \rctxone}
  {\reset \rctxzero \clocc\rel \reset\rctxone}
\end{mathpar}
Again, we consider only the restriction of $\clocc\rel$ to closed
contexts.

\section{Environmental Relations for the Relaxed Semantics}
\label{s:general}

In this section, we define an environmental bisimilarity which
characterizes the contextual equivalence
of~\cite{Biernacki-Lenglet:FOSSACS12,Biernacki-Lenglet:FLOPS12}, where
stuck terms can be observed.

\subsection{Contextual Equivalence}

We recall the definition of contextual equivalence $\ctxequiv$ for the
relaxed semantics (given in~\cite{Biernacki-Lenglet:FOSSACS12}).
\begin{definition}
\label{d:context}
  For all $\tmzero$, $\tmone$ be terms. We write $\tmzero \ctxequiv
  \tmone$ if for all $\cctx$ such that $\inctx \cctx \tmzero$ and
  $\inctx \cctx \tmone$ are closed, the following hold:
  \begin{itemize}
  \item $\inctx \cctx \tmzero \evalcbv \valzero$ implies $\inctx
    \cctx \tmone \evalcbv \valone$;
  \item $\inctx \cctx \tmzero \evalcbv \tmzero'$, where $\tmzero'$ is stuck,
    implies $\inctx \cctx \tmone \evalcbv \tmone'$, with $\tmone'$ stuck as
    well;
  \end{itemize}
  and conversely for $\inctx \cctx \tmone$.
\end{definition}
The definition is simpler when using the following context
lemma~\cite{Milner:TCS77} (for a proof see Section~3.4
in~\cite{Biernacki-Lenglet:FOSSACS12}). Instead of testing with
general, closing contexts, we can close the terms with values and then
put them in evaluation contexts.

\begin{lemma}[Context Lemma]
  \label{l:context-lemma}
  We have $\tmzero \ctxequiv \tmone$ iff for all closed
  contexts~$\rctx$ and for all substitutions $\sigma$ (mapping
  variables to closed values) such that $\tmzero\sigma$ and
  $\tmone\sigma$ are closed, the following hold:
  \begin{itemize}
  \item $\inctx \rctx {\tmzero\sigma} \evalcbv \valzero$ implies
    $\inctx \rctx {\tmone\sigma} \evalcbv \valone$;
  \item $\inctx \rctx {\tmzero\sigma} \evalcbv \tmzero'$, where
    $\tmzero'$ is stuck, implies $\inctx \rctx {\tmone\sigma} \evalcbv
    \tmone'$, with $\tmone'$ stuck as well;
  \end{itemize}
  and conversely for $\inctx \rctx {\tmone\sigma}$.
\end{lemma}
In~\cite{Biernacki-Lenglet:FOSSACS12}, we prove that $\ctxequiv$
satisfies all the axioms of CPS equivalence except for $\shift \vark
{\app \vark \tm} \eqKH \tm$ (provided $\vark \notin \fv \tm$): indeed,
$\shift \vark {\app \vark \tm}$ is stuck, but $\tm$ may evaluate to a
value. Conversely, some contextually equivalent terms are not CPS
equivalent, like Turing's and Church's call-by-value fixed point
combinators. Similarly, two arbitrary diverging terms are related by
$\ctxequiv$, but not necessarily by $\cpsequiv$.

\subsection{Definition of Environmental Bisimulation and Basic Properties}

Environmental bisimulations use an environment $\env$ to accumulate
knowledge about two tested terms. For the
$\lambda$-calculus~\cite{Sangiorgi-al:TOPLAS11}, $\env$ records the
values $(\valzero, \valone)$ the tested terms reduce to, if they
exist. We can then compare $\valzero$ and $\valone$ at any time by
passing them arguments built from $\env$. In $\lamshift$, we have to
consider stuck terms as well; therefore, environments may also contain
pairs of stuck terms, and we can test those by building pure contexts
from $\env$.

Formally, an environment $\env$ is a relation on normal forms which
relates values with values and stuck terms with stuck terms; e.g., the
identity environment $\Id$ is $\{(\tm, \tm) \mmid \tm \mbox{ is a
  normal form}\}$. An environmental relation $\erel$ is a set of
environments $\env$, and triples $(\env, \tmzero, \tmone)$, where
$\tmzero$ and $\tmone$ are closed. We write $\tmzero \ierel \erel \env
\tmone$ as a shorthand for $(\env, \tmzero, \tmone) \in \erel$;
roughly, it means that we test $\tmzero$ and $\tmone$ with the
knowledge~$\env$. The \emph{open extension} of $\erel$, written
$\open\erel$, is defined as follows: if $\vect \varx=\fv \tmzero \cup
\fv\tmone$\footnote{Given a metavariable $m$, we write $\vect m$ for a
  set of entities denoted by $m$.}, then we write $\tmzero
\open{\ierel \erel \env} \tmone$ if $\lam {\vect \varx} \tmzero \ierel
\erel \env \lam{\vect \varx} \tmone$.

\begin{definition}
  \label{d:env-bisim}
  A relation $\erel$ is an environmental bisimulation if
  \begin{enumerate}
  \item $\tmzero \ierel \erel \env \tmone$ implies: \label{e:tm}
    \begin{enumerate}
    \item if $\tmzero \redcbv \tmzero'$, then $\tmone \clocbv \tmone'$
      and $\tmzero' \ierel \erel \env \tmone'$; \label{e:tmtau}
    \item if $\tmzero=\valzero$, then $\tmone \clocbv
      \valone$ and $\env \cup \{(\valzero, \valone)\} \in
      \erel$; \label{e:tmval}
    \item if $\tmzero$ is stuck, then $\tmone \clocbv \tmone'$
      with $\tmone'$ stuck, and $\env \cup \{(\tmzero,
      \tmone')\} \in \erel$; \label{e:tmstuck}
    \item the converse of the above conditions on $\tmone$;
    \end{enumerate}
  \item $\env \in \erel$ implies: \label{e:env}
    \begin{enumerate}
    \item if $\lam \varx \tmzero \mathrel\env \lam \varx \tmone$ and $\valzero
      \cloctc \env \valone$, then $\subst \tmzero \varx \valzero \ierel \erel \env \subst
      \tmone \varx \valone$;\label{e:envv}
    \item if $\inctx \ctxzero {\shift \vark \tmzero} \mathrel\env \inctx \ctxone
      {\shift \vark \tmone}$ and $\ctx'_0 \clocc \env \ctx'_1$, then $\reset
      {\subst \tmzero \vark {\lam \varx {\reset {\inctx {\ctx'_0}{\inctx
                \ctxzero \varx}}}}} \ierel \erel \env \reset {\subst \tmone
        \vark {\lam \varx {\reset {\inctx {\ctx'_1}{\inctx \ctxone \varx}}}}}$
      for a fresh $\varx$.\label{e:envs}
    \end{enumerate}
  \end{enumerate}
\end{definition}

\emph{Environmental bisimilarity}, written $\bisim$, is the largest
environmental bisimulation. To prove that two terms $\tmzero$ and
$\tmone$ are equivalent, we want to relate them without any predefined
knowledge, i.e., we want to prove that $\tmzero \ierel \bisim
\emptyset \tmone$ holds; we also write $\empbisim$ for $\ierel \bisim
\emptyset$.

The first part of the definition makes the bisimulation game explicit
for $\tmzero$, $\tmone$, while the second part focuses on environments
$\env$. If $\tmzero$ is a normal form, then~$\tmone$ has to evaluate
to a normal form of the same kind, and we extend the environment with
the newly acquired knowledge. We then compare values in $\env$
(clause~(\ref{e:envv})) by applying them to arguments built from
$\env$, as in the
$\lambda$-calculus~\cite{Sangiorgi-al:TOPLAS11}. Similarly, we test
stuck terms in $\env$ by putting them within contexts
$\reset{\ctx'_0}$, $\reset{\ctx'_1}$ built from $\env$
(clause~(\ref{e:envs})) to trigger the capture. This reminds the way
we test values and stuck terms with applicative
bisimilarity~\cite{Biernacki-Lenglet:FOSSACS12}, except that
applicative bisimilarity tests both values or stuck terms with the
same argument or context. Using different entities (as in
Definition~\ref{d:env-bisim}) makes bisimulation proofs harder, but it
simplifies the proof of congruence of the environmental bisimilarity.

\begin{example}
  \label{ex:reset-beta}
  We have $\reset{\app {\lamp \varx \tmzero}{\reset \tmone}} \empbisim
  \app{\lamp \varx {\reset \tmzero}}{\reset \tmone}$, because the relation $\erel =
  \{ (\emptyset,\reset{\app {\lamp \varx \tm}{\reset {\tm'}}}, \app{\lamp \varx
    {\reset \tm}}{\reset {\tm'}}), (\emptyset,\reset{\app {\lamp \varx
      \tm} \val}, \app{\lamp \varx {\reset \tm}} \val)\} \cup \{(\env, \tm, \tm)
  \mmid \env \subseteq \mathord{\Id} \} \cup \{ \env \mmid \env \subseteq
  \mathord{\Id} \}$ is a bisimulation. Indeed, if $\reset
  {\tm'}$ evaluates to $\val$, then $\reset{\app {\lamp \varx \tm}{\reset
      {\tm'}}} \clocbv \reset{\app {\lamp \varx \tm} \val}$ and $\app{\lamp
    \varx {\reset \tm}}{\reset {\tm'}} \clocbv \app{\lamp \varx {\reset \tm}}
  \val$, which both reduce to $\reset {\subst \tm \varx \val}$.
\end{example}

As usual with environmental relations, the candidate relation $\erel$
in the above example could be made simpler with the help of up-to
techniques. 

Definition \ref{d:env-bisim} is written in the small-step style, because each
reduction step from $\tmzero$ has to be matched by $\tmone$. In the big-step
style, we are concerned only with evaluations to normal forms.
\begin{definition}
  \label{d:bs}
  A relation $\erel$ is a big-step environmental bi\-si\-mu\-la\-tion
  if $\tmzero \ierel \erel \env \tmone$ implies:
  \begin{enumerate}
  \item $\tmzero \ierel \erel \env \tmone$ implies: \label{e:bs-tm}
    \begin{enumerate}
    \item if $\tmzero \clocbv \valzero$, then $\tmone \clocbv
      \valone$ and $\env \cup \{(\valzero, \valone)\} \in
      \erel$;\label{e:bs-tmval}
    \item if $\tmzero \clocbv \tmzero'$ with $\tmzero'$ stuck, then
      $\tmone \clocbv \tmone'$, $\tmone'$ stuck, and $\env \cup
      \{(\tmzero', \tmone')\} \in \erel$;\label{e:bs-tmstuck} 
    \item the converse of the above conditions on $\tmone$;
    \end{enumerate}
  \item $\env \in \erel$ implies: \label{e:bs-env}
    \begin{enumerate}
    \item if $\lam \varx \tmzero \mathrel\env \lam \varx \tmone$ and $\valzero
      \cloctc \env \valone$, then $\subst \tmzero \varx \valzero \ierel \erel \env \subst
      \tmone \varx \valone$;\label{e:bs-envv}
    \item if $\inctx \ctxzero {\shift \vark \tmzero} \mathrel\env \inctx \ctxone
      {\shift \vark \tmone}$ and $\ctx'_0 \clocc \env \ctx'_1$, then $\reset
      {\subst \tmzero \vark {\lam \varx {\reset {\inctx {\ctx'_0}{\inctx
                \ctxzero \varx}}}}} \ierel \erel \env \reset {\subst \tmone
        \vark {\lam \varx {\reset {\inctx {\ctx'_1}{\inctx \ctxone \varx}}}}}$
      for a fresh $\varx$.\label{e:bs-envs}
    \end{enumerate}
  \end{enumerate}
\end{definition}

\begin{lemma}
  If $\erel$ is a big-step environmental bisimulation, then $\erel
  \subseteq \mathord{\bisim}$.
\end{lemma}
Big-step relations can be more convenient to use when we know the
result of the evaluation, as in Example~\ref{ex:reset-beta}, or as in
the following one.
\begin{example}
  \label{ex:reset-reset}
  We have $\reset{\reset \tm} \empbisim \reset \tm$. Indeed, we can
  show that $\reset{\reset \tm} \clocbv \val$ iff $\reset{\tm} \clocbv
  \val$, therefore $\{(\emptyset, \reset {\reset \tm}, \reset \tm)\}
  \cup \{(\env, \tm, \tm) \mmid \env \subseteq \mathord{\Id} \} \cup
  \{ \env \mmid \env \subseteq \mathord{\Id} \}$ is a big-step
  environmental bisimulation.
\end{example}
We use the following results in the rest of the paper.
\begin{lemma}[Weakening]
  \label{l:weakening}
  If $\tmzero \ebisim \tmone$ and $\env' \subseteq \env$ then $\tmzero
  \ierel \bisim {\env'} \tmone$.
\end{lemma}

A smaller environment is a weaker constraint, because we can build less
arguments and contexts to test the normal forms in $\env$. The proof is as
in~\cite{Sangiorgi-al:TOPLAS11}. Lemma~\ref{l:redcbv-in-empbisim} states that
reduction (and therefore, evaluation) is included in~$\empbisim$.
\begin{lemma}
  \label{l:redcbv-in-empbisim}
  If $\tmzero \redcbv \tmzero'$, then $\tmzero \empbisim \tmzero'$.
\end{lemma}

\subsection{Soundness and Completeness}
\label{s:soundness-completeness}

We now prove soundness and completeness of $\empbisim$
w.r.t. contextual equivalence. Because the proofs follow the same
steps as for the $\lambda$-calculus~\cite{Sangiorgi-al:TOPLAS11}, we
only give here the main lemmas and sketch their proofs. The complete
proofs can be found in Appendix~\ref{a:general}. First, we need some
basic up-to techniques, namely up-to environment (which allows bigger
environments in the bisimulation clauses) and up-to bisimilarity
(which allows for limited uses of $\empbisim$ in the bisimulation
clauses), whose definitions and proofs of soundness are
classic~\cite{Sangiorgi-al:TOPLAS11}.

With these tools, we can prove that $\empbisim$ is sound and complete
w.r.t. contextual equivalence. For a relation~$\rel$ on terms, we
write $\rvals \rel$ for its restriction to closed normal forms. The
first step consists in proving congruence for normal forms, and also
for any terms but only w.r.t. evaluation contexts.

\begin{lemma}
  \label{l:cong-val-main}
  Let $\tmzero$, $\tmone$ be normal forms. If $\tmzero \ebisim
  \tmone$, then $\inctx \cctx \tmzero \ierel \bisim \env \inctx \cctx
  \tmone$.
\end{lemma}

\begin{lemma}
  \label{l:cong-e-ctx-main}
  If $\tmzero \ebisim \tmone$, then $\inctx \rctx \tmzero \ierel
  \bisim \env \inctx \rctx \tmone$.
\end{lemma}
Lemmas~\ref{l:cong-val-main} and~\ref{l:cong-e-ctx-main} are proved
simultaneously by showing that, for any environmental bisimulation
$\erely$, the relation
\begin{multline*}
  \erel = \{(\rvals {\cloct \env}, \inctx \rctxzero \tmzero, \inctx
  \rctxone \tmone) \mmid \tmzero \ierel \erely \env \tmone, \rctxzero
  \clocc \env \rctxone\} \\ \cup \{(\rvals {\cloct \env}, \tmzero,
  \tmone) \mmid \env \in \mathord{\erely}, \tmzero \cloctc \env \tmone \}
  \cup \{\rvals{\cloct \env} \mmid \env \in \mathord{\erely} \}
\end{multline*}
is a bisimulation up-to environment. Informally, the elements of the
first set of $\erel$ reduce to elements of the second set of $\erel$,
and we then prove the bisimulation property for these elements by
induction on $\tmzero \mathrel{\cloctc \env} \tmone$. We can then
prove the main congruence lemma.
\begin{lemma}
  \label{l:cong-empbisim-main}
  $\tmzero \empbisim \tmone$ implies $\inctx \cctx \tmzero \ierel
  \bisim {\rvals {\cloct \empbisim}} \inctx \cctx \tmone$.
\end{lemma}
We show that $\{ (\rvals {\cloct \empbisim}, \tmzero, \tmone) \mmid
\tmzero \cloctc \empbisim \tmone\} \cup \{ \rvals {\cloct \empbisim}
\}$ is a bisimulation up-to bisimilarity by induction on $\tmzero
\cloctc \empbisim \tmone$. By weakening (Lemma \ref{l:weakening}), we
can deduce from Lemma \ref{l:cong-empbisim-main} that $\empbisim$ is a
congruence, and therefore is sound w.r.t. $\ctxequiv$.

\begin{corollary}[Soundness]
  \label{c:soundness-main}
  We have $\mathord{\empbisim} \subseteq \mathord{\ctxequiv}$.
\end{corollary}
The relation $\empbisim$ is also complete w.r.t. contextual
equivalence.
\begin{theorem}[Completeness]
  We have $\mathord{\ctxequiv} \subseteq \mathord{\empbisim}$.
\end{theorem}
The proof is by showing that $\{(\rvals \ctxequiv, \tmzero, \tmone)
\mmid \tmzero \ctxequiv \tmone\} \cup \{ \rvals \ctxequiv \}$ is a
big-step bisimulation, using Lemma \ref{l:context-lemma} as an
alternate definition for $\ctxequiv$.

\subsection{Bisimulation up to context}
\label{s:up-to-context}

Equivalence proofs based on environmental bisimilarity can be
simplified by using up-to techniques, such as up to reduction, up to
expansion, and up to context~\cite{Sangiorgi-al:TOPLAS11}. We only
discuss the last, since the first two can be defined and proved sound
in $\lamshift$ without issues. Bisimulations up to context may factor
out a common context from the tested terms. Formally, we define the
context closure of $\erel$, written $\clocr\erel$, as follows: we have
$\tmzero \clocr{\ierel \erel \env} \tmone$ if
\begin{itemize}
\item either $\tmzero = \inctx \rctxzero {\tmzero'}$, $\tmone = \inctx
  \rctxone {\tmone'}$, $\tmzero' \ierel \erel \env \tmone'$, and
  $\rctxzero \clocc \env \rctxone$;
\item or $\tmzero \cloctc \env \tmone$.
\end{itemize}
Note that terms $\tmzero'$ and $\tmone'$ (related by $\ierel \erel
\env$) can be put into evaluation contexts only, while normal forms
(related by $\env$) can be put in any contexts. This restriction to
evaluation contexts in the first case is usual in the definition of
up-to context techniques for environmental
relations~\cite{Sangiorgi-al:TOPLAS11,Sumii:TCS10,Sato-Sumii:APLAS09,Pierard-Sumii:LICS12}.

\begin{definition}
  \label{d:utctx}
  A relation $\erel$ is an environmental bisimulation up to context if
  \begin{enumerate}
  \item $\tmzero \ierel \erel \env \tmone$ implies: \label{e:utctx-tm}
    \begin{enumerate}
    \item if $\tmzero \redcbv \tmzero'$, then $\tmone \clocbv \tmone'$
      and $\tmzero' \clocr{\ierel \erel \env}
      \tmone'$; \label{e:utctx-tmtau}
    \item if $\tmzero=\valzero$, then $\tmone \clocbv \valone$ and
      $\env \cup \{(\valzero, \valone)\} \subseteq \mathord{\rvals{\cloct{\env'}}}$ for some
      $\env' \in \erel$ ; \label{e:utctx-tmval}
    \item if $\tmzero$ is stuck, then $\tmone \clocbv \tmone'$
      with $\tmone'$ stuck, and $\env \cup \{(\tmzero,
      \tmone')\} \subseteq \mathord{\rvals{\cloct{\env'}}}$ for some
      $\env' \in \erel$; \label{e:utctx-tmstuck}
    \item the converse of the above conditions on $\tmone$;
    \end{enumerate}
  \item $\env \in \erel$ implies: \label{e:utctx-env}
    \begin{enumerate}
    \item if $\lam \varx \tmzero \mathrel\env \lam \varx \tmone$ and 
       $\valzero \cloctc\env \valone$, then $\subst \tmzero \varx \valzero \clocr{\ierel \erel \env}
      \subst \tmone \varx \valone$; \label{e:utctx-envv}
    \item if $\inctx \ctxzero {\shift \vark \tmzero} \mathrel\env \inctx \ctxone
      {\shift \vark \tmone}$ and $\ctx'_0 \clocc \env \ctx'_1$, then $\reset
      {\subst \tmzero \vark {\lam \varx {\reset {\inctx {\ctx'_0}{\inctx
                \ctxzero \varx}}}}} \clocr{\ierel \erel \env} \reset {\subst
        \tmone \vark {\lam \varx {\reset {\inctx {\ctx'_1}{\inctx \ctxone
                \varx}}}}}$ for a fresh $\varx$.\label{e:utctx-envs}
    \end{enumerate}
  \end{enumerate}
\end{definition}

\begin{lemma}
  \label{l:utctx-soundness}
  If $\erel$ is an environmental bisimulation up to context, then
  $\erel \subseteq \mathord{\bisim}$.
\end{lemma}
The soundness proof is the same as
in~\cite{Sangiorgi-al:TOPLAS11}. While this definition is enough to
simplify proofs in the $\lambda$-calculus case, it is not that helpful
in $\lamshift$, because of the restriction to evaluation contexts
(first item of the definition of $\clocr \erel$). In the
$\lambda$-calculus, when a term $\tm$ reduces
within an evaluation context, the context is not affected, hence
Definition \ref{d:utctx} is enough to help proving interesting
equivalences. It is not the case in $\lamshift$, as (a part of) the
evaluation context can be captured.

Indeed, suppose we want to construct a candidate relation $\erel$ to
prove the $\beta_\Omega$ axiom, i.e., $\inctx \ctx \tm$ is equivalent
to $\app{\lamp \varx {\inctx \ctx \varx}} \tm$, assuming $\varx \notin
\fv \ctx$. The problematic case is when $\tm$ is a stuck term $\inctx
\ctxzero {\shift \vark \tmzero}$; we have to add the stuck terms
$\app{\lamp \varx {\inctx \ctx \varx}}{\inctx \ctxzero {\shift \vark
    \tmzero}}$ and $\inctx \ctx {\inctx \ctxzero {\shift \vark
    \tmzero}}$ to an environment $\env$ of $\erel$. For $\erel$ to be
a bisimulation, we then have to prove that for all $\ctxone \clocc\env
\ctxtwo$, we have $\reset{\subst {\tmzero} \vark {\lam y
    {\reset{\inctx \ctxone {\app {\lamp \varx {\inctx \ctx
              \varx}}{\inctx \ctxzero y}}}}}} \ierel \erel \env
\reset{\subst {\tmzero} \vark {\lam y {\reset{\inctx \ctxtwo {\inctx
          \ctx {\inctx \ctxzero y}}}}}}$. At this point, we would like
to use the up-to context technique, because the subterms $\app {\lamp
  \varx {\inctx \ctx \varx}}{\inctx \ctxzero y}$ and $\inctx \ctx
{\inctx \ctxzero y}$ are similar to the terms we want to relate (they
can be written $\app{\lamp \varx {\inctx \ctx \varx}}{\tm''}$ and
$\inctx \ctx {\tm''}$ with $\tm''=\inctx \ctxzero y$). However, we
have at best $\reset{\subst {\tmzero} \vark {\lam y {\reset{\inctx
        \ctxone {\app {\lamp \varx {\inctx \ctx \varx}}{\inctx
            \ctxzero y}}}}}} \cloct{\open{\ierel \erel \env}}
\reset{\subst {\tmzero} \vark {\lam y {\reset{\inctx \ctxtwo {\inctx
          \ctx {\inctx \ctxzero y}}}}}}$ (and not $\clocr{\ierel \erel
  \env}$), because (i) $\app {\lamp \varx {\inctx \ctx \varx}}{\inctx
  \ctxzero y}$ and $\inctx \ctx {\inctx \ctxzero y}$ are open terms,
and (ii) $\tmzero$ can be any term, so $\app {\lamp \varx {\inctx \ctx
    \varx}}{\inctx \ctxzero y}$ and $\inctx \ctx {\inctx \ctxzero y}$
can be put in any context, not necessarily in an evaluation
one. Therefore, Definition~\ref{d:utctx} cannot help there.

Problem (ii) could be somewhat dealt with in the particular case of
the $\beta_\Omega$ axiom by changing clause (\ref{e:utctx-envs}) of
Definition~\ref{d:utctx} into
\begin{enumerate}
    \item[\textit{(b)}] if $\inctx \ctxzero {\shift \vark \tmzero} \mathrel\env \inctx \ctxone
      {\shift \vark \tmone}$ and $\ctx'_0 \clocc{\ierel \erel \env} \ctx'_1$, then $\reset
      {\subst \tmzero \vark {\lam \varx {\reset {\inctx {\ctx'_0}{\inctx
                \ctxzero \varx}}}}} \cloct{\ierel \erel \env} \reset {\subst
        \tmone \vark {\lam \varx {\reset {\inctx {\ctx'_1}{\inctx \ctxone
                \varx}}}}}$ for a fresh $\varx$.
\end{enumerate}
and similarly for clause (\ref{e:utctx-envv}). In plain text, we build
the testing contexts $\ctx'_0$, $\ctx'_1$ from $\ierel \erel \env$
(instead of $\env$), and the resulting terms have to be in
$\cloct{\ierel \erel \env}$ (without any evaluation context
restriction). The resulting notion of bisimulation up to context is
sound. The new clause would be more difficult to establish in general
than the original one (of Definition~\ref{d:utctx}), because it tests
more pairs of contexts. However, for the $\beta_\Omega$ axiom, we
would have to prove that for all $\ctxone \cloct{\ierel \erel \env}
\ctxtwo$, $\reset{\subst {\tmzero} \vark {\lam y {\reset{\inctx
        \ctxone {\app {\lamp \varx {\inctx \ctx \varx}}{\inctx
            \ctxzero y}}}}}} \cloct{\ierel \erel \env} \reset{\subst
  {\tmzero} \vark {\lam y {\reset{\inctx \ctxtwo {\inctx \ctx {\inctx
            \ctxzero y}}}}}}$ holds; it would be easy, except $\app
       {\lamp \varx {\inctx \ctx \varx}}{\inctx \ctxzero y}$ and
       $\inctx \ctx {\inctx \ctxzero y}$ are open terms (problem (i)).

Problem (i) seems harder to fix, because for $\app {\lamp \varx
  {\inctx \ctx \varx}}{\inctx \ctxzero y} \open{\ierel \erel \env}
\inctx \ctx {\inctx \ctxzero y}$ to hold, we must have $\app {\lamp
  \varx {\inctx \ctx \varx}}{\inctx \ctxzero \valzero} \ierel \erel
\env \inctx \ctx {\inctx \ctxzero \valone}$ for all $\valzero
\cloct\env \valone$. Because $\ctxzero$ can be anything, it means that
we must have $\app {\lamp \varx {\inctx \ctx \varx}}{\tmzero'} \ierel
\erel \env \inctx \ctx {\tmone'}$ with $\tmzero' \cloct\env \tmone'$;
$\tmzero'$ and $\tmone'$ are plugged in different contexts, therefore
bisimulation up to context (which factors out only a common context)
cannot help us there; a new kind of up-to technique is required.

The $\beta_\Omega$ axiom example suggests that we need more powerful
up-to techniques for environmental bisimilarity for delimited control;
we leave these potential improvements as a future work. Note that we
do not have such issues with up-to techniques for normal form
bisimilarity: it relates open terms without having to replace their
free variables, and normal form bisimulation up to context is not
restricted to evaluation contexts only. But even if environmental
bisimulation up to context is not as helpful as wished, it still
simplifies equivalence proofs, as we can see with the next example.

\begin{example}
  In \cite{Danvy-Filinski:DIKU89}, a variant of Turing's call-by-value fixed
  point combinators using shift and reset has been proposed. Let $\theta = \lam
  {\varx y}{\app y {\lamp z {\app {\app {\app \varx \varx} y} z}}}$. We prove
  that $\tmzero = \app \theta \theta$ is bisimilar to its variant $\tmone =
  \reset {\app \theta {\shift \vark {\app \vark \vark}}}$. Let $\theta' = \lam
  \varx {\reset{\app \theta \varx}}$, $\valzero = \lam y {\app y {\lamp z {\app
        {\app {\app \theta \theta} y} z}}}$, and $\valone = \lam y {\app y
    {\lamp z {\app {\app {\app {\theta'}{\theta'}} y} z}}}$. We define $\env$
  inductively such that $\valzero \mathrel\env \valone$, and if $\valzero'
  \cloctc\env \valone'$, then $\lam z {\app {\app {\app \theta
        \theta}{\valzero'}} z} \mathrel\env \lam z {\app {\app {\app
        {\theta'}{\theta'}}{\valone'}} z}$. Then $\erel = \{ (\env, \tmzero,
  \tmone), (\env, \tmzero, \app{\theta'}{\theta'}), \env \}$ is a (big-step)
  bisimulation up to context. Indeed, we have $\tmzero \evalcbv \valzero$,
  $\tmone \evalcbv \valone$, and $\app{\theta'}{\theta'} \evalcbv \valone$,
  therefore clause (\ref{e:utctx-tmval}) of Definition
  \ref{d:utctx} is checked for both pairs. We now check clause
  (\ref{e:utctx-envv}), first for $\valzero \mathrel\env
  \valone$. For all $\valzero' \cloctc\env \valone'$, we have $\app
  {\valzero'}{\lamp z {\app {\app {\app \theta \theta}{\valzero'}} z}} \cloctc
  \env \app {\valone'}{\lamp z {\app {\app {\app{ \theta'} {\theta'}}{\valone'}}
      z}}$ (because $\lam z {\app {\app {\app \theta \theta}{\valzero'}} z}
  \mathrel\env \lam z {\app {\app {\app{ \theta'} {\theta'}}{\valone'}} z}$),
  hence the result holds. Next, let $\lam z {\app {\app {\app \theta
        \theta}{\valzero'}} z} \mathrel\env \lam z {\app {\app {\app{ \theta'}
        {\theta'}}{\valone'}} z}$ (with $\valzero' \cloctc \env \valone'$), and
  let $\valzero'' \cloctc\env \valone''$. We have to check that $\app {\app
    {\app \theta \theta}{\valzero'}}{\valzero''} \clocr{\ierel \erel \env} \app
  {\app {\app{ \theta'} {\theta'}}{\valone'}}{\valone''}$, which is true,
  because $\app \theta \theta \ierel \erel \env \app {\theta'}{\theta'}$, and
  $\apctx{\apctx \mtctx {\valzero'}}{\valzero''} \clocc\env \apctx{\apctx \mtctx
    {\valone'}}{\valone''}$.
\end{example}

\section{Environmental Relations for the Original Semantics}
\label{s:programs}

The original CPS semantics for shift and reset~\cite{Danvy-Filinski:LFP90} as
well as the corresponding reduction semantics~\cite{Biernacka-al:LMCS05} assume
that terms can be considered as programs to be executed, only when surrounded by
a top-level reset. In this section, we present a CPS-compatible bisimulation
theory that takes such a requirement into account. In this section, we call
\emph{programs}, ranged over by $\prg$, terms of the form $\reset\tm$.

\subsection{Contextual Equivalence}

To reflect the fact that terms are executed within an enclosing reset,
the contextual equivalence we consider in this section tests terms in
contexts of the form $\reset \cctx$ only. Because programs cannot
reduce to stuck terms, the only possible observable action is
evaluation to values. We therefore define contextual equivalence for
programs as follows.
\begin{definition}
  \label{d:context-p}
  Let $\tmzero$, $\tmone$ be terms. We write $\tmzero \ctxequivp
  \tmone$ if for all $\cctx$ such that $\reset{\inctx \cctx \tmzero}$
  and $\reset{\inctx \cctx \tmone}$ are closed, $\reset{\inctx \cctx
    \tmzero} \evalcbv \valzero$ implies $\reset{\inctx \cctx \tmone}
  \evalcbv \valone$, and conversely for $\reset{\inctx \cctx \tmone}$.
\end{definition}
Note that $\ctxequivp$ is defined on all terms, not just programs. It
is easy to check that $\ctxequiv$ is more discriminative than
$\ctxequivp$. We will see in Section~\ref{s:examples-p} that this
inclusion is in fact strict.

\begin{lemma}
  We have $\mathord{\ctxequiv} \subseteq \mathord{\ctxequivp}$.
\end{lemma}

\subsection{Definition and Properties}

We now propose a definition of environmental bisimulation adapted to
programs (but defined on all terms, like $\ctxequivp$). Because stuck
terms are no longer observed, environments $\env$ henceforth relate
only values. Similarly, we write $\rval\rel$ for the restriction of a
relation $\rel$ on terms to pairs of closed values.

\begin{definition}
  \label{d:env-bisim-p}
  A relation $\erel$ is an environmental bisimulation for programs if
  \begin{enumerate}
  \item if $\tmzero \ierel \erel \env \tmone$ and $\tmzero$ and
    $\tmone$ are not both programs, then for all $\ctxzero \clocc\env
    \ctxone$, we have $\reset{\inctx \ctxzero \tmzero} \ierel \erel
    \env \reset{\inctx \ctxone \tmone}$; \label{e:tm-p}
  \item if $\prgzero \ierel \erel \env \prgone$ \label{e:prg-p}
    \begin{enumerate}
    \item if $\prgzero \redcbv \prgzero'$, then $\prgone \clocbv
      \prgone'$ and $\prgzero' \ierel \erel \env
      \prgone'$; \label{e:prgtau-p}
    \item if $\prgzero \redcbv \valzero$, then $\prgone \clocbv
      \valone$, and $\{(\valzero, \valone)\} \cup \env \in
      \erel$; \label{e:prgval-p}
    \item the converse of the above conditions on $\prgone$;
    \end{enumerate}
  \item for all $\env \in \erel$, if $\lam \varx \tmzero 
    \mathrel{\env} \lam \varx \tmone$ and $\valzero \cloctc\env
    \valone$, then $\subst \tmzero \varx \valzero \ierel \erel \env
    \subst \tmone \varx \valone$. \label{e:env-p}
  \end{enumerate}
\end{definition}
Environmental bisimilarity for programs, written $\bisimp$, is the
largest environmental bisimulation for programs. As before, the
relation $\ierel \bisimp \emptyset$, also written $\empbisimp$, is
candidate to characterize $\ctxequivp$.

Clauses (\ref{e:prg-p}) and (\ref{e:env-p}) of Definition
\ref{d:env-bisim-p} deal with programs and environment in a classical
way (as in plain $\lambda$-calculus). The problematic case is when
relating terms $\tmzero$ and $\tmone$ that are not both programs
(clause (\ref{e:tm-p})). Indeed, one of them may be stuck, and
therefore we have to test them within some contexts $\reset\ctxzero$,
$\reset\ctxone$ (built from $\env$) to potentially trigger a capture
that otherwise would not happen. We cannot require both terms to be
stuck, as in clause (\ref{e:envs}) of Definition~\ref{d:env-bisim},
because a stuck term can be equivalent to a term free from control
effect. E.g., we will see that $\val \empbisimp \shift \vark {\app
  \vark \val}$, provided that $\vark \notin \fv\val$.

\begin{example}
  Suppose we want to prove $\reset{\app {\lamp \varx \tmzero}{\reset
      \tmone}} \empbisimp \app{\lamp \varx {\reset \tmzero}}{\reset
    \tmone}$ (as in Example~\ref{ex:reset-beta}). Because $\app{\lamp
    \varx {\reset \tmzero}}{\reset \tmone}$ is not a program, we have
  to put both terms into a context first: we have to change the
  candidate relation of Example~\ref{ex:reset-beta} into $\erel = \{
  (\emptyset,\reset{\app {\lamp \varx \tmzero}{\reset \tmone}},
  \app{\lamp \varx {\reset \tmzero}}{\reset \tmone})\} \cup
  \{(\emptyset,\reset{\inctx \ctx{\reset{\app {\lamp \varx
          \tmzero}{\reset \tmone}}}}, \reset{\inctx \ctx{\app{\lamp
        \varx {\reset \tmzero}}{\reset \tmone}}})\} \cup
  \{(\emptyset,\reset{\inctx \ctx{\reset{\app {\lamp \varx \tmzero}
        \val}}}, \reset{\inctx \ctx{\app{\lamp \varx {\reset \tmzero}}
      \val}})\} \cup \{(\env, \tm, \tm) \mmid \env \subseteq
  \mathord{\Id} \} \cup \{ \env \mmid \env \subseteq \mathord{\Id}
  \}$. In contrast, to prove $\reset{\reset \tm} \empbisimp \reset
  \tm$, we do not have to change the candidate relation of
  Example~\ref{ex:reset-reset}, since both terms are programs.
\end{example}
We can give a definition of big-step bisimulation by removing clause
(\ref{e:prgtau-p}) and changing $\redcbv$ into $\clocbv$ in clause
(\ref{e:prgval-p}). Lemmas~\ref{l:weakening}
and~\ref{l:redcbv-in-empbisim} can also be extended to $\bisimp$ and
$\empbisimp$. %
The next lemma shows that $\empbisim$ is more discriminative than
$\empbisimp$.
\begin{lemma}
  \label{l:empbisim-in-empbisimp}
  We have $\mathord{\empbisim} \subseteq \mathord{\empbisimp}$.
\end{lemma}

A consequence of
Lemma~\ref{l:empbisim-in-empbisimp} is that we can use
Definition~\ref{d:env-bisim} as a proof technique for
$\empbisimp$. E.g., we have directly $\reset{\app {\lamp \varx
    \tmzero}{\reset \tmone}} \empbisimp \app{\lamp \varx {\reset
    \tmzero}}{\reset \tmone}$, because $\reset{\app {\lamp \varx
    \tmzero}{\reset \tmone}} \empbisim \app{\lamp \varx {\reset
    \tmzero}}{\reset \tmone}$.

\review{3.  sect 4.2 paragraph just after lemma 12.  I could not
  completely understand the reasoning here in the first two lines the
  relation with suffix "v" ($\simeq^v$) is used but the conclusion (in
  the third line) the relatin without "v" is used ($\simeq$).  The
  latter is a superset of the former, so the reasoning itself is OK,
  but then why did you choose this setting?  Also, this paragraph
  seems to show that the relation $\{(\simeq^v, t_0, t_1) | t_0 \simeq
  t_1\} \cup ...$ satisfies the condition in Def. 10, but why it
  implies Lemma 12 ?  It seems that the following relation is more
  relevant here: $\{(\emptyset, t_0, t_1) | t_0 \simeq t_1\} \cup
  \{\simeq^v\}$ I recommend either (1) to add more explanations so
  that this paragraph indeed becomes a proof sketch of Lemma 12, or
  (2) to remove the first three lines of this paragraph.}\serguei{I
  would go with option (1), but maybe option (2) is preferable for
  space reasons }\darek{We will have some space after removing Remark
  1, so maybe (1) is doable?}\serguei{Maybe, we will see when we are
  done addressing the other comments}\darek{I am OK with option (2)
  given that we need to save space.}

\subsection{Soundness and Completeness}

We sketch the proofs of soundness and completeness of $\empbisimp$
w.r.t.~$\ctxequivp$; see Appendix~\ref{a:programs} for the complete
proofs. The soundness proof follows the same scheme as in
Section~\ref{s:soundness-completeness}, with some necessary
adjustments. As before, we need up-to environment and up-to
bisimilarity techniques to prove the following lemmas. 

\begin{lemma}
  \label{l:cong-val-p-main}
  If $\valzero \ebisimp \valone$, then
  $\inctx \cctx \valzero \ierel \bisimp \env \inctx \cctx \valone$.
\end{lemma}

\begin{lemma}
  \label{l:cong-e-ctx-p-main}
  If $\tmzero \ebisimp \tmone$, then
  $\inctx \rctx \tmzero \ierel \bisimp \env \inctx \rctx \tmone$.
\end{lemma}

We prove Lemmas~\ref{l:cong-val-p-main} and \ref{l:cong-e-ctx-p-main} by showing
that a relation similar to the relation~$\erel$ defined in
Section~\ref{s:soundness-completeness} is a bisimulation up to environment. 
We then want to prove the main congruence lemma, akin to
Lemma~\ref{l:cong-empbisim-main}, by showing that $\erely \mathord= \{
(\rval {\cloct \empbisimp}, \tmzero, \tmone) \mmid \tmzero \cloctc
\empbisimp \tmone\} \cup \{ \rval {\cloct \empbisimp} \}$ is a
bisimulation up to bisimilarity. However, we can no longer proceed by
induction on $\tmzero \cloctc \empbisimp \tmone$, as for
Lemma~\ref{l:cong-empbisim-main}. Indeed, if $\prgzero = \reset
\tmzero$, $\prgone = \reset \tmone$ with $\tmzero \cloctc \empbisimp
\tmone$, and if $\tmzero$ is a stuck term, then $\prgzero$ reduces to
some term, but the induction hypothesis does not tell us anything
about $\tmone$. To circumvent this, we decompose related programs into
related subcomponents.

\begin{lemma}
  \label{l:decompose-prg-main}
  If $\prgzero \cloctc\empbisimp \prgone$, then either $\prgzero \empbisimp
  \prgone$, or one of the following holds:
  \begin{itemize}
  \item $\prgzero = \reset \valzero$;
  \item $\prgzero = \inctx \rctxzero {\reset{\inctx \ctxzero
      \tmzero}}$, $\prgone = \inctx \rctxone {\reset{\inctx \ctxone
      \tmone}}$ , $\rctxzero \clocc \empbisimp \rctxone$, $\ctxzero
    \clocc\empbisimp \ctxone$, $\tmzero \empbisimp \tmone$ and
    $\tmzero \redcbv \tmzero'$ or $\tmzero$ is stuck;
  \item $\prgzero = \inctx \rctxzero {\reset{\inctx \ctxzero
      \redexzero}}$, $\prgone = \inctx \rctxone {\reset{\inctx \ctxone
      \tmone}}$ , $\rctxzero \clocc \empbisimp \rctxone$, $\ctxzero
    \clocc\empbisimp \ctxone$, $\redexzero \cloctc\empbisimp \tmone$
    but $\redexzero \not\empbisimp \tmone$.
  \end{itemize}
\end{lemma}
Lemma~\ref{l:decompose-prg-main} generalizes
Lemma~\ref{l:unique-decomp} to related programs: we know $\prgzero$
can be decomposed into  contexts $\rctx$,
$\reset\ctx$, and a redex $\redex$, and we relate these subterms to
$\prgone$. We can then prove that $\erely$ (defined above) is a
bisimulation up to bisimilarity, by showing that, in each case
described by Lemma~\ref{l:decompose-prg-main}, $\prgzero$ and
$\prgone$ reduce to terms related by $\erely$. From this, we deduce
$\empbisimp$ is a congruence, and is sound w.r.t. $\ctxequivp$.

\begin{lemma}
  \label{l:cong-empbisim-main-p}
  $\tmzero \empbisimp \tmone$ implies $\inctx \cctx \tmzero \ierel
  \bisimp {\rval {\cloct \empbisimp}} \inctx \cctx \tmone$.
\end{lemma}

\begin{corollary}[Soundness]
  \label{c:soundness-p-main}
  We have $\mathord{\empbisimp} \subseteq \mathord{\ctxequivp}$.
\end{corollary}

\begin{remark}
  \label{r:applicative}
  Following the ideas behind Definition~\ref{d:env-bisim-p}, one can
  define an applicative bisimilarity $\bis$ for programs. However,
  proving that $\bis$ is sound seems more complex than for
  $\empbisimp$. We remind that the soundness proof of an applicative
  bisimilarity consists in showing that a relation called the
  \emph{Howe's closure} $\cloh\bis$ is an applicative bisimulation. To
  this end, we need a version of Lemma~\ref{l:decompose-prg-main} for
  $\cloh\bis$. However, $\cloh\bis$ is inductively defined as the
  smallest congruence which contains $\bis$ and satisfies
  $\mathord{\cloh\bis \bis} \subseteq \mathord{\cloh\bis}$ (1), and
  condition (1) makes it difficult to write a decomposition lemma for
  $\cloh\bis$ similar to Lemma~\ref{l:decompose-prg-main}.
\end{remark}

We prove completeness of $\empbisimp$ by showing that the relation
$\ctxequivep$, defined below, coincides with $\ctxequivp$ and $\empbisimp$. By
doing so, we also prove a context lemma for $\ctxequivp$.
\begin{definition}
  Let $\tmzero$, $\tmone$ be closed terms. We write $\tmzero
  \ctxequivep \tmone$ if for all closed~$\rctx$, $\reset{\inctx \rctx
    \tmzero} \evalcbv \valzero$ implies $\reset{\inctx \rctx \tmone}
  \evalcbv \valone$, and conversely for $\reset{\inctx \rctx \tmone}$.
\end{definition}
By definition, we have $\mathord{\ctxequivp} \subseteq
\mathord{\ctxequivep}$. With the same proof technique as in
Section~\ref{s:soundness-completeness}, we prove the following lemma.
\begin{lemma}[Completeness]
  \label{l:completeness-p-main}
  We have $\mathord{\ctxequivep} \subseteq
  \mathord{\empbisimp}$. 
\end{lemma}
With Lemma~\ref{l:completeness-p-main} and Corollary
\ref{c:soundness-p-main}, we have $\mathord{\ctxequivp} \subseteq
\mathord{\ctxequivep} \subseteq \mathord{\empbisimp} \subseteq
\mathord{\ctxequivp}$. Defining up-to context for programs is
possible, with the same limitations as in Section~\ref{s:up-to-context}.

\subsection{Examples}
\label{s:examples-p}

We illustrate the differences between $\empbisim$ and $\empbisimp$, by
giving some examples of terms related by $\empbisimp$, but not by
$\empbisim$. First, note that $\empbisimp$ relates non-terminating
terms with stuck non-terminating terms.
\begin{lemma}
  \label{l:omega-stuck}
  We have $\Omega \empbisimp \shift \vark \Omega$.
\end{lemma}
The relation $\{(\emptyset, \Omega, \shift \vark \Omega), (\emptyset,
\reset{\inctx \ctx \Omega}, \reset{\inctx \ctx {\shift \vark
    \Omega}}), (\emptyset, \reset{\inctx \ctx \Omega}, \reset{\Omega})
\}$ is a bisimulation for programs. Lemma~\ref{l:omega-stuck} does not hold
with $\empbisim$ because $\Omega$ is not stuck.

As wished, $\empbisimp$ satisfies the only axiom
of~\cite{Kameyama-Hasegawa:ICFP03} not satisfied by $\empbisim$.
\begin{lemma}
  \label{l:skkt}
  If $\vark \notin \fv \tm$, then $\tm \open\empbisimp \shift \vark {\app \vark \tm}$.
\end{lemma}
We sketch the proof for $\tm$ closed; for the general case, see Appendix
\ref{a:skkt-open}. We prove that $\{(\emptyset, \tm, \shift \vark {\app \vark
  \tm}), (\emptyset, \reset{\inctx \ctx \tm}, \reset{\inctx \ctx {\shift \vark
    {\app \vark \tm}}})\} \cup \empbisim$ is a bisimulation for
programs. Indeed, we have $\reset{\inctx \ctx {\shift \vark {\app \vark \tm}}}
\redcbv \reset{\app {\lamp \varx {\reset{\inctx \ctx \varx}}} \tm}$, and because
$\empbisim$ verifies the $\beta_\Omega$ axiom ($\empbisim$ is complete, and
$\ctxequiv$ verifies the $\beta_\Omega$ axiom
\cite{Biernacki-Lenglet:FOSSACS12}), we know that $\reset{\app {\lamp \varx
    {\reset{\inctx \ctx \varx}}} \tm} \empbisim \reset{\reset{\inctx \ctx \tm}}$
holds. From Example~\ref{ex:reset-reset}, we have $\reset{\reset{\inctx \ctx
    \tm}} \empbisim \reset{\inctx \ctx \tm}$, therefore we have $\reset{\inctx
  \ctx {\shift \vark {\app \vark \tm}}} \empbisim \reset{\inctx \ctx \tm}$.

Consequently, $\open\empbisimp$ is complete w.r.t. $\cpsequiv$.

\begin{corollary}
  We have $\mathord{\cpsequiv} \subseteq \mathord{\open\empbisimp}$.
\end{corollary}
As a result, we can use $\cpsequiv$ (restricted to closed terms) as a proof
technique for~$\empbisimp$. E.g., the following equivalence can be derived from
the axioms~\cite{Kameyama-Hasegawa:ICFP03}.
\begin{lemma}
  \label{l:s-tail}
  If $\vark \notin \fv\tmone$, then $\app{\lamp \varx{\shift \vark
      \tmzero}} \tmone \empbisimp \shift\vark {\appp {\lamp \varx
      \tmzero} \tmone}$.
\end{lemma}
This equivalence does not hold with $\empbisim$, because the term on
the right is stuck, but the term on the left may not evaluate to a
stuck term (if $\tmone$ does not terminate). We can generalize this
result as follows, again by using $\cpsequiv$.
\begin{lemma}
  \label{l:generalized-s-tail}
  If $\vark \notin \fv\tmone$ and $\varx \notin \fv \ctx$, then we
  have $\app{\lamp \varx{\inctx \ctx {\shift \vark \tmzero}}} \tmone
  \empbisimp \inctx \ctx {\shift\vark {\appp {\lamp \varx \tmzero}
      \tmone}}$.
\end{lemma}
Proving Lemma~\ref{l:skkt} without the $\beta_\Omega$ axiom and
Lemmas~\ref{l:s-tail} and~\ref{l:generalized-s-tail} without
$\cpsequiv$ requires complex candidate relations (see the proof of
Lemma~\ref{l:s-tail} in Appendix~\ref{a:s-tail}), because of the lack
of powerful enough up-to techniques.

\serguei{An example of terminating terms that are not related by
  $\cpsequiv$ would be nice}\darek{I guess we have it (e.g., Lemmas 20
  and 21 adjusted so that $t_1$ in one term is the Curry fixed-point
  combinator and in the other it is the Turing fixed-point
  combinator). Perhaps we want to save such examples for the journal
  version.}

\section{Conclusion}
\label{s:conclusion}

We propose sound and complete environmental bisimilarities for two
variants of the semantics of $\lamshift$. For the semantics of
Section~\ref{s:general}, we now have several bisimilarities, each with
its own merit. Normal form
bisimilarity~\cite{Biernacki-Lenglet:FLOPS12} and its up-to techniques
leads to minimal proof obligations, however it is not complete, and
distinguishes very simple equivalent terms (see Proposition~1
in~\cite{Biernacki-Lenglet:FLOPS12}). Applicative bisimilarity
\cite{Biernacki-Lenglet:FOSSACS12} is complete but sometimes requires
complex bisimulation proofs (e.g., for the $\beta_\Omega$
axiom). Environmental bisimilarity~$\empbisim$
(Definition~\ref{d:env-bisim}) is also complete, can be difficult to
use, but this difficulty can be mitigated with up-to
techniques. However, bisimulation up to context is not as helpful as
we could hope (see Section~\ref{s:up-to-context}), because we have to
manipulate open terms (problem (i)), and the context closure of an
environmental relation is restricted to evaluation contexts (problem
(ii)). As a result, proving the $\beta_\Omega$ axiom is more difficult
with environmental than with applicative bisimilarity. We believe
dealing with problem (i) requires new up-to techniques to be
developed, and lifting the evaluation context restriction (problem
(ii)) would benefit not only for $\lamshift$, but also for process
calculi with passivation~\cite{Pierard-Sumii:LICS12}; we leave this as
a future work.

In contrast, we do not have as many options when considering the
semantics of Section~\ref{s:programs} (where terms are evaluated
within a top-level reset). The environmental bisimilarity of this
paper $\empbisimp$ (Definition~\ref{d:env-bisim-p}) is the first to be
sound and complete w.r.t. Definition~\ref{d:context-p}. As argued
in~\cite{Biernacki-Lenglet:FLOPS12} (Section 3.2), normal form
bisimilarity cannot be defined on programs without introducing extra
quantifications (which defeats the purpose of normal form
bisimilarity). Applicative bisimilarity could be defined for programs,
but proving its soundness would require a new technique, since the
usual one (Howe's method) does not seem to apply (see
Remark~\ref{r:applicative}). This confirms that environmental
bisimilarity is more flexible than applicative
bisimilarity~\cite{Koutavas-al:ENTCS11}. However, we would like to
simplify the quantification over contexts in clause (\ref{e:tm-p}) of
Definition~\ref{d:env-bisim-p}, so we look for sub-classes of terms
where this quantification is not mandatory.

Other future works include the study of the behavioral theory of other
delimited control operators, like the dynamic ones (e.g., {\it
  control} and {\it prompt}~\cite{Felleisen:POPL88} or {\it shift$_0$}
and {\it reset$_0$}~\cite{Danvy-Filinski:DIKU89}), but also of
abortive control operators, such as {\it callcc}, for which no sound and
complete bisimilarity has been defined so far.

\subsubsection*{Acknowledgments}
\label{sec:acks}
We thank Ma{\l}gorzata Biernacka and the anonymous referees for many
helpful comments on the presentation of this work.

\bibliographystyle{abbrv}
\bibliography{mybib}

\newpage

\appendix

\section{Soundness and Completeness for the Relaxed Semantics}
\label{a:general}

In bisimulation up-to environment, one can use bigger environments
that the ones needed by Definition \ref{d:env-bisim}. As a result,
instead of making the environment grow at each bisimulation step, we
can directly use the largest possible environment.
\begin{definition}
  \label{d:utenv}
  An environmental relation $\erel$ is an environmental bisimulation
  up to environment if
  \begin{enumerate}
  \item $\tmzero \ierel \erel \env \tmone$ implies: \label{e:utenv-tm}
    \begin{enumerate}
    \item if $\tmzero \redcbv \tmzero'$, then $\tmone \clocbv \tmone'$
      and $\tmzero' \ierel \erel {\env'} \tmone'$ for some $\env'$
      such that $\env \subseteq \env'$; \label{e:utenv-tmtau}
    \item if $\tmzero$ is a value $\valzero$, then $\tmone \clocbv
      \valone$ and $\env' \in \erel$ for some $\env'$ such that $\env
      \cup \{(\valzero, \valone)\} \subseteq
      \env'$; \label{e:utenv-tmval}
    \item if $\tmzero$ is a stuck term, then $\tmone \clocbv \tmone'$
      where $\tmone'$ is a stuck term and $\env' \in \erel$ for some
      $\env'$ such that $\env \cup \{(\tmzero, \tmone')\} \subseteq
      \env'$; \label{e:utenv-tmstuck}
    \item the converse of the above conditions on $\tmone$;
    \end{enumerate}
  \item $\env \in \erel$ implies: \label{e:utenv-env}
    \begin{enumerate}
    \item for all $(\lam \varx \tmzero, \lam \varx \tmone) \in \env$,
      for all $(\valzero, \valone) \in \mathord{\cloctc \env}$, we
      have $\subst \tmzero \varx \valzero \ierel \erel {\env'} \subst
      \tmone \varx \valone$ for some $\env \subseteq
      \env'$; \label{e:utenv-envv}
    \item for all $(\inctx \ctxzero {\shift \vark \tmzero}, \inctx
      \ctxone {\shift \vark \tmone}) \in \env$, for all $(\ctx'_0,
      \ctx'_1) \in \mathord{\clocc \env}$, we have $$\reset {\subst
        \tmzero \vark {\lam \varx {\reset {\inctx {\ctx'_0}{\inctx
                \ctxzero \varx}}}}} \ierel \erel {\env'} \reset
              {\subst \tmone \vark {\lam \varx {\reset {\inctx
                      {\ctx'_1}{\inctx \ctxone \varx}}}}}$$ for a
              fresh $\varx$ and some $\env \subseteq
              \env'$.\label{e:utenv-envs}
    \end{enumerate}
  \end{enumerate}
\end{definition}

\begin{lemma}
  If $\erel$ is an environmental bisimulation up to environment, then
  $\erel \subseteq \bisim$.
\end{lemma}

Next, we define bisimulation up-to bisimilarity, where we can compose
with $\empbisim$ to simplify the definition of candidate relations by
factoring out useless bisimilar terms.
\begin{definition}
  \label{d:utbisim}
  An environmental relation $\erel$ is an environmental bisimulation
  up to bisimilarity if
  \begin{enumerate}
  \item $\tmzero \ierel \erel \env \tmone$ implies: \label{e:utbisim-tm}
    \begin{enumerate}
    \item if $\tmzero \redcbv \tmzero'$, then $\tmone \clocbv \tmone'$
      and $\tmzero' \ierel \erel \env \empbisim
      \tmone'$; \label{e:utbisim-tmtau}
    \item if $\tmzero$ is a value $\valzero$, then $\tmone \clocbv
      \valone$ and $\env \cup \{(\valzero, \valone')\} \in \erel$ for
      some $\valone' \empbisim \valone$; \label{e:utbisim-tmval}
    \item if $\tmzero$ is a stuck term, then $\tmone \clocbv \tmone'$
      where $\tmone'$ is a stuck term and $\env \cup \{(\tmzero,
      \tmone'')\} \in \erel$ for some stuck term $\tmone''$ such that
      $\tmone'' \empbisim \tmone'$; \label{e:utbisim-tmstuck}
    \item the converse of the above conditions on $\tmone$;
    \end{enumerate}
  \item $\env \in \erel$ implies: \label{e:utbisim-env}
    \begin{enumerate}
    \item for all $(\lam \varx \tmzero, \lam \varx \tmone) \in \env$,
      for all $(\valzero, \valone) \in \mathord{\cloctc \env}$, we
      have $\subst \tmzero \varx \valzero \empbisim \ierel \erel \env
      \empbisim \subst \tmone \varx \valone$; \label{e:utbisim-envv}
    \item for all $(\inctx \ctxzero {\shift \vark \tmzero}, \inctx
      \ctxone {\shift \vark \tmone}) \in \env$, for all $(\ctx'_0,
      \ctx'_1) \in \mathord{\clocc \env}$, we have $$\reset {\subst
        \tmzero \vark {\lam \varx {\reset {\inctx {\ctx'_0}{\inctx
                \ctxzero \varx}}}}} \empbisim \ierel \erel \env
      \empbisim \reset {\subst \tmone \vark {\lam \varx {\reset
            {\inctx {\ctx'_1}{\inctx \ctxone \varx}}}}}$$ for a fresh
      $\varx$.\label{e:utbisim-envs}
    \end{enumerate}
  \end{enumerate}
\end{definition}

\begin{lemma}
  If $\erel$ is an environmental bisimulation up to bisimilarity, then
  $\erel \subseteq \,\bisim$.
\end{lemma}
As usual with up-to bisimilarity with small-step relations, we cannot compose on
the left-hand side of $\erel$ in clause (\ref{e:utbisim-tm}) of
Definition~\ref{d:utbisim}.

\begin{lemma}
  \label{l:cong-is-subst}
  Let $\rel$ be a relation on closed terms. If $\tmzero \cloct\rel
  \tmone$ (where $\tmzero$ and $\tmone$ are potentially open terms)
  and $\valzero \rvals{\cloct\rel} \valone$, then $\subst \tmzero
  \varx \valzero \cloct\rel \subst \tmone \varx \valone$.
\end{lemma}

\begin{proof}
  We proceed by induction on $\tmzero \cloct\rel \tmone$. Suppose
  $\tmzero = \tmone = \varx$. We have $\valzero \rvals{\cloct\rel}
  \valone$ as wished. The result is also easy if $\tmzero = \tmone = y
  \neq \varx$. Suppose $\tmzero \rel \tmone$. Because $\rel$ is
  defined on closed terms only, we have $\subst \tmzero \varx \valzero
  = \tmzero \rel \tmone = \subst \tmone \varx \valone$. The remaining
  induction cases are straightforward.

\end{proof}

\begin{lemma}
  \label{l:prop-cloct-env}
  Let $\env$ be an environment (i.e., a relation on closed values and
  closed stuck terms only). Suppose $\tmzero \cloctc\env \tmone$. If
  $\tmzero$ is a value, then so is $\tmone$, and if $\tmzero$ is a
  stuck term, then so is $\tmone$.
\end{lemma}

\begin{proof}
  The first item is straightforward by case analysis on $\tmzero
  \cloctc\env \tmone$ (and using the fact that $\env$ relates values
  only with values), and the second item is straightforward by
  induction on $\tmzero \cloctc\env \tmone$ (and using the fact that
  $\env$ relates stuck terms only with stuck terms).

\end{proof}

\begin{lemma}
  \label{l:cong-val}
  For all $\env$ and normal forms $\tmzero$, $\tmone$, if $\tmzero
  \ebisim \tmone$, then $\inctx \cctx \tmzero \ierel \bisim
          {\rvals{\cloct \env}} \inctx \cctx \tmone$.
\end{lemma}

\begin{lemma}
  \label{l:cong-e-ctx}
  For all $\env$, if $\tmzero \ebisim \tmone$, then $\inctx \rctx
  \tmzero \ierel \bisim {\rvals{\cloct \env}} \inctx \rctx \tmone$.
\end{lemma}
We prove Lemmas~\ref{l:cong-val} and~\ref{l:cong-e-ctx}
simultaneously. Let $\erely$ be an environmental bisimulation. We
define
\begin{align*}
  \erel & = \erel_1 \cup \erel_2 \cup \{\rvals{\cloct \env} \mmid \env \in \mathord{\erely} \}
  \\
  \erel_1 & = \{(\rvals {\cloct \env}, \inctx \rctxzero \tmzero, \inctx \rctxone
  \tmone) \mmid \tmzero \ierel \erely \env \tmone, \rctxzero \clocc \env \rctxone\}
  \\
  \erel_2 & = \{(\rvals {\cloct \env}, \tmzero, \tmone) \mmid \env \in \mathord{\erely},
  \tmzero \cloctc \env \tmone \}
\end{align*}
In $\erel_2$, we build the closed terms $(\tmzero, \tmone)$ out of
pairs of values or pair of stuck terms. We first prove a preliminary
lemma about $\erel$.

\begin{lemma}
  \label{l:erel}
  Let $\env \in \mathord{\erely}$.
  \begin{itemize}
  \item If $\lam \varx \tmzero \cloctc \env \lam \varx \tmone$ and
    $\valzero \rvals{\cloct \env} \valone$ then $\subst \tmzero \varx
    \valzero \ierel \erel {\rvals {\cloct \env}} \subst \tmone \varx
    \valone$.
  \item If $\inctx \ctxzero {\shift \vark \tmzero} \cloctc \env \inctx
    \ctxone {\shift \vark \tmone}$ and $\ctx'_0 \clocc \env
    \ctx'_1$, then $\reset{\subst \tmzero \vark {\lam \varx {\reset
          {\inctx {\ctx'_0}{\inctx \ctxzero \varx}}}}} \ierel \erel
            {\rvals {\cloct \env}} \reset{\subst \tmone \vark {\lam
                \varx {\reset {\inctx {\ctx'_1}{\inctx \ctxone
                      \varx}}}}}$.
  \end{itemize}
\end{lemma}

\begin{proof}
  For the first item, we proceed by case analysis on $\lam \varx
  \tmzero \cloctc \env \lam \varx \tmone$. If $\lam \varx \tmzero \env
  \lam \varx \tmone$, then since $\erely$ is an environmental
  bisimulation, we have $\subst \tmzero \varx \valzero \ierel \erely
  \env \subst \tmone \varx \valone$, which implies $\subst \tmzero
  \varx \valzero \ierel \erel {\rvals {\cloct \env}} \subst \tmone
  \varx \valone$ (more precisely, the terms are in $\erel_1$).

  If $\tmzero \cloct \env \tmone$ with $\fv\tmzero \cup \fv\tmone
  \subseteq \{\varx\}$, then we have $\subst \tmzero \varx \valzero
  \cloct \env \subst \tmone \varx \valone$ by
  Lemma~\ref{l:cong-is-subst}. In fact, we have $\subst \tmzero \varx
  \valzero \cloctc \env \subst \tmone \varx \valone$, so we have
  $\subst \tmzero \varx \valzero \ierel \erel {\rvals {\cloct \env}}
  \subst \tmone \varx \valone$ (more precisely, the terms are in
  $\erel_2$).\\

  For the second item, we proceed by induction on $\inctx \ctxzero
  {\shift \vark \tmzero} \cloctc \env \inctx \ctxone {\shift \vark
    \tmone}$. If $\inctx \ctxzero {\shift \vark \tmzero} \env \inctx
  \ctxone {\shift \vark \tmone}$, then because $\erely$ is an
  environmental bisimulation, we have $\reset{\subst \tmzero \vark
    {\lam \varx {\reset {\inctx {\ctx'_0}{\inctx \ctxzero \varx}}}}}
  \ierel \erely \env \reset{\subst \tmone \vark {\lam \varx {\reset
        {\inctx {\ctx'_1}{\inctx \ctxone \varx}}}}}$, which is equivalent to 
  $\reset{\subst \tmzero \vark {\lam \varx {\reset {\inctx
          {\ctx'_0}{\inctx \ctxzero \varx}}}}} \ierel \erel {\rvals
    {\cloct \env}} \reset{\subst \tmone \vark {\lam \varx {\reset
        {\inctx {\ctx'_1}{\inctx \ctxone \varx}}}}}$ (the terms
  are in $\erel_1$).

  Suppose $\ctxzero = \ctxone = \mtctx$ and $\tmzero \cloct \env
  \tmone$ with $\fv \tmzero \cup \fv \tmone \subseteq \{\vark\}$. From
  $\ctx'_0 \clocc \env \ctx'_1$, we deduce $\lam \varx
  \reset{\inctx{\ctx'_0} \varx} \cloctc\env \lam \varx
  \reset{\inctx{\ctx'_1} \varx}$. We
  have $\reset{\subst \tmzero \vark {\lam \varx {\reset {\inctx
          {\ctx'_0}{\inctx \ctxzero \varx}}}}} \cloctc \env
  \reset{\subst \tmone \vark {\lam \varx {\reset {\inctx
          {\ctx'_1}{\inctx \ctxone \varx}}}}}$,  by
  Lemma~\ref{l:cong-is-subst}, hence the result
  holds (the terms are in~$\erel_2$).

  Suppose $\ctxzero = \vctx \valzero {\ctx''_0}$ and $\ctxone = \vctx
  \valone {\ctx''_1}$ with $\valzero \cloctc \env \valone$ and $\inctx
          {\ctx''_0}{\shift \vark \tmzero} \clocc \env
          \inctx{\ctx''_1}{\shift \vark \tmone}$. From $\valzero
          \cloctc \env \valone$ and $\ctx'_0 \clocc \env \ctx'_1$,
          we deduce $\inctx{\ctx'_1}{\vctx \valzero \mtctx} \clocc
          \env \inctx{\valone'}{\vctx \valone \mtctx}$. Then
          $\reset{\subst \tmzero \vark {\lam \varx {\reset {\inctx
                  {\ctx'_0}{\app \valzero {\inctx {\ctx''_0}
                      \varx}}}}}} \ierel \erel {\rvals {\cloct \env}}
          \reset{\subst \tmone \vark {\lam \varx {\reset {\inctx
                  {\ctx'_1}{\app \valone {\inctx {\ctx''_1}
                      \varx}}}}}}$ by the induction hypothesis, i.e.,
          $\reset{\subst \tmzero \vark {\lam \varx {\reset {\inctx
                  {\ctx'_0}{\inctx \ctxzero \varx}}}}} \ierel \erel
                {\rvals {\cloct \env}} \reset{\subst \tmone \vark
                  {\lam \varx {\reset {\inctx {\ctx'_1}{\inctx
                          \ctxone \varx}}}}}$, as wished. The case
                $\ctxzero = \apctx {\ctx''_0}{\tmzero'}$ and
                $\ctxone = \apctx{\ctx''_1}{\tmone'}$ is similar.

\end{proof}

We now prove Lemmas~\ref{l:cong-val} and~\ref{l:cong-e-ctx} by showing
that $\erel$ is a bisimulation up to environment.
\begin{proof}
  We first prove the bisimulation for the elements in $\erel_2$ (for
  these, we do not need the ``up to environment''). Let $\tmzero
  \cloctc \env \tmone$, with $\env \in \mathord{\erely}$. Clause
  \ref{e:tmval} (resp. \ref{e:tmstuck}) is easy: if $\tmzero$ is a
  value (resp. a stuck term), then so is $\tmone$ (cf. Lemma
  \ref{l:prop-cloct-env}), and we have $\rvals{\cloct \env} \cup
  \{(\tmzero, \tmone) \} = \rvals{\cloct \env} \in \erel$. For clause
  \ref{e:tmtau}, we proceed by induction on $\tmzero \cloctc \env
  \tmone$.

  Suppose $\tmzero = \app{\tmzero^1}{\tmzero^2}$ and $\tmone =
  \app{\tmone^1}{\tmone^2}$ with $\tmzero^1 \cloctc \env \tmone^1$ and
  $\tmzero^2 \cloctc \env \tmone^2$. We have three cases to consider.
  \begin{itemize}
  \item Assume $\tmzero^1 \redcbv {\tmzero^1}'$, so that $\tmzero
    \redcbv \app{{\tmzero^1}'}{\tmzero^2}$. By the induction
    hypothesis, there exists ${\tmone^1}'$ such that $\tmone^1 \clocbv
    {\tmone^1}'$ and ${\tmzero^1}' \ierel \erel {\rvals{\cloctc\env}}
    {\tmone^1}'$. From $\tmzero^2 \cloctc \env \tmone^2$ and
    ${\tmzero^1}' \ierel \erel {\rvals{\cloctc\env}} {\tmone^1}'$, we
    can deduce $\app {{\tmzero^1}'}{\tmzero^2} \ierel \erel
    {\rvals{\cloctc\env}} \app{{\tmone^1}'}{\tmone^2}$ by definition
    of $\erel$. We also have $\tmone \clocbv
    \app{{\tmone^1}'}{\tmone^2}$, hence the result holds.
  \item Assume $\tmzero^1=\valzero$ and $\tmzero^2 \redcbv
    {\tmzero^2}'$, so that $\tmzero \redcbv \app
    \valzero{{\tmzero^2}'}$. Then $\tmone^1$ is also a value $\valone$
    according to Lemma \ref{l:prop-cloct-env}. By the induction
    hypothesis, there exists ${\tmone^2}'$ such that $\tmone^2 \clocbv
    {\tmone^2}'$ and ${\tmzero^2}' \ierel \erel {\rvals{\cloctc\env}}
    {\tmone^2}'$. From $\valzero \cloctc \env \valone^2$ and
    ${\tmzero^2}' \ierel \erel {\rvals{\cloctc\env}} {\tmone^2}'$, we
    can deduce $\app \valzero {{\tmzero^2}'} \ierel \erel
    {\rvals{\cloctc\env}} \app \valone {{\tmone^2}'}$ by definition of
    $\erel$. We also have $\tmone \clocbv \app \valone {{\tmone^2}'}$,
    hence the result holds.
  \item Assume $\tmzero^1 = \lam \varx {\tmzero'}$ and $\tmzero^2 =
    \valzero$, so that $\tmzero \redcbv \subst {\tmzero'} \varx
    \valzero$. By Lemma~\ref{l:prop-cloct-env}, $\tmone^1$ is a value
    $\lam \varx {\tmone'}$ and $\tmone^2$ is a value $\valone$. We
    have $\tmone \redcbv \subst {\tmzero'} \varx \valzero$, and by
    Lemma~\ref{l:erel}, we have $\subst {\tmzero'} \varx \valzero
    \ierel \erel {\rvals{\cloctc\env}} \subst {\tmone'} \varx
    \valone$, hence the result holds.
  \end{itemize}
  Suppose $\tmzero = \reset{\tmzero'}$, $\tmone = \reset{\tmone'}$
  with $\tmzero' \cloctc\env \tmone'$. We have two possibilities.
  \begin{itemize}
  \item Assume $\tmzero' \redcbv \tmzero''$, so that $\tmzero \redcbv
    \reset{\tmzero''}$. By the induction hypothesis, there exists
    $\tmone''$ such that $\tmone' \clocbv \tmone''$ and $\tmzero''
    \ierel \erel {\rvals{\cloctc\env}} \tmone''$. By definition of
    $\erel$, we have $\reset{\tmzero''} \ierel \erel
           {\rvals{\cloctc\env}} \reset{\tmone''}$, and furthermore
           $\tmzero \clocbv \reset{\tmone''}$, we therefore have the
           required result.
  \item Assume $\tmzero' = \inctx \ctxzero {\shift \vark
    {\tmzero''}}$, so that $\tmzero \redcbv \reset{\subst {\tmzero''}
    \vark {\lam \varx {\reset{\inctx \ctxzero \varx}}}}$. By
    Lemma~\ref{l:prop-cloct-env}, $\tmone'$ is a stuck term $\inctx
    \ctxone {\shift \vark {\tmone''}}$, therefore $\tmone \redcbv
    \reset{\subst{\tmone''} \vark {\lam \varx {\reset{\inctx \ctxone
            \varx}}}}$. We have $\reset{\subst {\tmzero''} \vark {\lam
        \varx {\reset{\inctx \ctxzero \varx}}}} \ierel \erel
          {\rvals{\cloctc\env}}\reset{\subst{\tmone''} \vark {\lam
              \varx {\reset{\inctx \ctxone \varx}}}}$ by
          Lemma~\ref{l:erel}, hence the result holds.
  \end{itemize}

  We now prove the bisimulation property (up to environment) for
  elements in $\erel_1$. Let $\inctx \rctxzero \tmzero \ierel \erel
  {\rvals{\cloct \env}} \inctx \rctxone \tmone$, so that $\tmzero
  \ierel \erely \env \tmone$ and $\rctxzero \clocc \env \rctxone$. If
  $\tmzero$ is a value $\valzero$, then because $\erely$ is a
  bisimulation, there exists $\valone$ such that $\tmone \clocbv
  \valone$ and $\env' =\env \cup \{(\tmzero, \tmone)\} \in
  \mathord{\erely}$. We then have $\inctx \rctxone \tmone \clocbv
  \inctx \rctxone \valone$, and the terms $\inctx \rctxzero \valzero$,
  $\inctx \rctxone \valone$ are similar to the one of $\erel_1$. We
  can prove the bisimulation property with $\inctx \rctxzero
  \valzero$, $\inctx \rctxone \valone$ the same way we did with the
  terms in $\erel_1$, except that we reason up to environment, because
  $\env \subseteq \env'$. The reasoning is similar if $\tmzero$ is a
  stuck term. Suppose $\tmzero$ is not a value nor a stuck term. There
  exists $\tmzero'$ such that $\tmzero \redcbv \tmzero'$, and so
  $\inctx \rctxzero \tmzero \redcbv \inctx \rctxzero
  {\tmzero'}$. Because $\erely$ is a bisimulation, there exists
  $\tmone'$ such that $\tmone \clocbv \tmone'$ and $\tmzero' \ierel
  \erely \env \tmone'$. We therefore have $\inctx \rctxone \tmone
  \clocbv \inctx \rctxone {\tmone'}$ with $\inctx \rctxzero {\tmzero'}
  \ierel \erel {\rvals{\cloct \env}} \inctx \rctxone {\tmone'}$, as
  wished.

  We now prove the clause~\ref{e:env} of the bisimulation. The only environments
  in $\erel$ are of the form $\rvals {\cloct \env}$. Let $\lam \varx \tmzero
  \rvals {\cloct \env} \lam \varx \tmone$ and $\valzero \rvals {\cloct \env}
  \valone$. By Lemma~\ref{l:erel}, we have $\subst \tmzero \varx \valzero \ierel
  \erel {\rvals {\cloct \env}} \subst \tmone \varx \valone$, hence the result
  holds. Similarly, if  $\inctx \ctxzero {\shift \vark \tmzero} \rvals{\cloct
    \env} \inctx \ctxone {\shift \vark \tmone}$ and $\ctx'_0 \clocc \env
  \ctx'_1$, then (by Lemma~\ref{l:erel}) we have $\reset{\subst \tmzero \vark
    {\lam \varx {\reset {\inctx {\ctx'_0}{\inctx \ctxzero \varx}}}}} \ierel
  \erel {\rvals {\cloct \env}} \reset{\subst \tmone \vark {\lam \varx {\reset
        {\inctx {\ctx'_1}{\inctx \ctxone \varx}}}}}$.

\end{proof}

\begin{lemma}
  \label{l:val}
  If $\lam \varx \tmzero \empbisim \lam \varx \tmone$, then $\subst
  \tmzero \varx \val \empbisim \subst \tmone \varx \val$.
\end{lemma}

\begin{proof}
  By clause \ref{e:tmval}, we have $\{(\lam \varx \tmzero, \lam \varx
  \tmone) \} \in \mathord{\empbisim}$. Let $\env = \{(\lam \varx
  \tmzero, \lam \varx \tmone) \}$. By clause \ref{e:envv}, for all
  $\val$, we have $\subst \tmzero \varx \val \ebisim \subst \tmone
  \varx \val$, therefore $\subst \tmzero \varx \val \empbisim \subst
  \tmone \varx \val$ holds by weakening (Lemma~\ref{l:weakening}).
\end{proof}

\begin{lemma}
  \label{l:shift}
  If $\inctx \ctxzero {\shift \vark \tmzero} \empbisim \inctx \ctxone
  {\shift \vark \tmone}$, then we have $\reset{\subst \tmzero \vark {\lam
      \varx {\reset{\inctx \ctx {\inctx \ctxzero \varx}}}}} \empbisim
  \reset{\subst \tmone \vark {\lam \varx {\reset{\inctx \ctx {\inctx
            \ctxone \varx}}}}}$.
\end{lemma}

\begin{proof}
  By clause \ref{e:tmstuck}, we know that $\{(\inctx \ctxzero {\shift
    \vark \tmzero}, \inctx \ctxone {\shift \vark \tmone}) \} \in
  \mathord{\empbisim}$. Let $\env = \{(\inctx \ctxzero {\shift \vark
    \tmzero}, \inctx \ctxone {\shift \vark \tmone}) \}$. By
  clause~\ref{e:envs}, we know that $\reset{\subst \tmzero \vark {\lam
      \varx {\reset{\inctx \ctx {\inctx \ctxzero \varx}}}}} \ebisim
  \reset{\subst \tmone \vark {\lam \varx {\reset{\inctx \ctx {\inctx
            \ctxone \varx}}}}}$, hence $\reset{\subst \tmzero \vark
    {\lam \varx {\reset{\inctx \ctx {\inctx \ctxzero \varx}}}}}
  \empbisim \reset{\subst \tmone \vark {\lam \varx {\reset{\inctx \ctx
          {\inctx \ctxone \varx}}}}}$ is true by weakening
  (Lemma~\ref{l:weakening}).
\end{proof}

\begin{lemma}
  \label{l:subst-val}
  If $\lam \varx \tmzero \cloctc \empbisim \lam \varx \tm \empbisim
  \lam \varx \tmone$ and $\valzero \cloctc \empbisim \val \empbisim
  \valone$ then $\subst \tmzero \varx \valzero \cloctc \empbisim
  \empbisim \subst \tmone \varx \valone$.
\end{lemma}

\begin{proof}
  We proceed by case analysis on $\lam \varx \tmzero \cloctc \empbisim
  \lam \varx \tm$.

  Suppose $\lam \varx \tmzero \empbisim \lam \varx \tm$. We have
  $\subst \tmzero \varx \valzero \cloctc \empbisim \subst \tmzero
  \varx \val$ by Lemma \ref{l:cong-is-subst}, $\subst \tmzero \varx
  \val \empbisim \subst \tm \varx \val$ by Lemma~\ref{l:val}, $\subst
  \tm \varx \val \empbisim \subst \tm \varx \valone$ by
  Lemma~\ref{l:cong-val}, and $\subst \tm \varx \valone \empbisim
  \subst \tmone \varx \valone$ by Lemma~\ref{l:val}. Finally, $\subst
  \tmzero \varx \valzero \cloctc \empbisim \empbisim \subst \tmone
  \varx \valone$ holds using transitivity of $\empbisim$.

  Suppose $\tmzero \cloct \empbisim \tm$ with $\fv \tmzero \cup \fv
  \tm \subseteq \{ \varx \}$. We have $\subst \tmzero \varx \valzero
  \cloctc\empbisim \subst \tm \varx \val$ by
  Lemma~\ref{l:cong-is-subst}, $\subst \tm \varx \val \empbisim \subst
  \tm \varx \valone$ by Lemma~\ref{l:cong-val}, and $\subst \tm \varx
  \valone \empbisim \subst \tmone \varx \valone$ by
  Lemma~\ref{l:val}. Finally, $\subst \tmzero \varx \valzero \cloctc
  \empbisim \empbisim \subst \tmone \varx \valone$ holds using
  transitivity of $\empbisim$.
\end{proof}

\begin{lemma}
  \label{l:subst-shift}
  If $\inctx \ctxzero {\shift \vark \tmzero} \cloctc \empbisim \inctx
  \ctx {\shift \vark \tm}$, $\inctx \ctx {\shift \vark \tm} \empbisim
  \inctx \ctxone {\shift \vark \tmone}$ and $\ctx'_0 \clocc\empbisim
  \ctx'_1$, then $\reset{\subst \tmzero \vark {\lam \varx
      {\reset{\inctx {\ctx'_0}{\inctx \ctxzero \varx}}}}} \cloctc
  \empbisim \empbisim \reset{\subst \tmone \vark {\lam \varx
      {\reset{\inctx {\ctx'_1}{\inctx \ctxone \varx}}}}}$.
\end{lemma}

\begin{proof}
  We start with proving that $\inctx \ctxzero {\shift \vark \tmzero}
  \cloctc \empbisim \inctx \ctx {\shift \vark \tm}$ and $\ctx'_0
  \clocc\empbisim \ctx'_1$ implies $\reset{\subst \tmzero \vark {\lam
      \varx {\reset{\inctx {\ctx'_0}{\inctx \ctxzero \varx}}}}}
  \cloctc \empbisim \empbisim \reset{\subst \tm \vark {\lam \varx
      {\reset{\inctx {\ctx'_1}{\inctx \ctx \varx}}}}}$. We proceed by
  induction on $\inctx \ctxzero {\shift \vark \tmzero} \cloctc
  \empbisim \inctx \ctx {\shift \vark \tm}$.

  Suppose $\inctx \ctxzero {\shift \vark \tmzero} \empbisim \inctx
  \ctx {\shift \vark \tm}$. From $\ctx'_0 \clocc\empbisim \ctx'_1$,
  we get $\lam \varx {\reset{\inctx {\ctx'_0}{\inctx \ctxzero
        \varx}}} \cloctc\empbisim \lam \varx {\reset{\inctx
      {\ctx'_1}{\inctx \ctxzero \varx}}}$. Then $\reset{\subst
    \tmzero \vark {\lam \varx {\reset{\inctx {\ctx'_0}{\inctx
            \ctxzero \varx}}}}} \cloctc \empbisim \reset{\subst
    \tmzero \vark {\lam \varx {\reset{\inctx {\ctx'_1}{\inctx
            \ctxzero \varx}}}}}$ holds by Lemma~\ref{l:cong-is-subst},
  and then $\reset{\subst \tmzero \vark {\lam \varx {\reset{\inctx
          {\ctx'_1}{\inctx \ctxzero \varx}}}}} \empbisim
  \reset{\subst \tm \vark {\lam \varx {\reset{\inctx {\ctx'_1}{\inctx
            \ctx \varx}}}}}$ holds by Lemma~\ref{l:shift}, hence the
  result holds.

  Suppose $\ctxzero=\ctx=\mtctx$ and $\tmzero \cloctc\empbisim \tm$
  with $\fv \tmzero \cup \fv \tm \subseteq \{ \vark \}$. From
  $\ctx'_0 \clocc\empbisim \ctx'_1$, we have $\lam \varx
  {\reset{\inctx {\ctx'_0} \varx}} \cloctc\empbisim \lam \varx
  {\reset{\inctx {\ctx'_1} \varx}}$. Then $\reset{\subst \tmzero
    \vark {\lam \varx {\reset{\inctx {\ctx'_0} \varx}}}}
  \cloctc\empbisim \reset{\subst \tm \vark {\lam \varx {\reset{\inctx
          {\ctx'_1} \varx}}}}$ by Lemma~\ref{l:cong-is-subst}, hence
  the result holds.

  Suppose $\inctx \ctxzero {\shift \vark \tmzero} = \app \valzero
  {\inctx {\ctx''_0}{\shift \vark \tmzero}}$, $\inctx \ctx {\shift
    \vark \tm} = \app \val {\inctx {\ctx''}{\shift \vark \tm}}$ with
  $\valzero \cloctc\empbisim \val$ and $\inctx {\ctx''_0}{\shift
    \vark \tmzero} \cloctc\empbisim \inctx {\ctx''}{\shift \vark
    \tm}$. From $\ctx'_0 \clocc\empbisim \ctx'_1$ and $\valzero
  \cloctc\empbisim \val$, it is the case that $\inctx{\ctx'_0}{\vctx
    \valzero \mtctx} \clocc\empbisim \inctx{\ctx'_1}{\vctx \val
    \mtctx}$. By the induction hypothesis, we obtain
  $$\reset{\subst \tmzero \vark {\lam \varx {\reset{\inctx
          {\ctx'_0}{\app \valzero {\inctx {\ctx''_0} \varx}}}}}}
  \cloctc \empbisim \empbisim \reset{\subst \tm \vark {\lam \varx
      {\reset{\inctx {\ctx'_1}{\app \val {\inctx \ctx \varx}}}}}},$$
  which means that $\reset{\subst \tmzero \vark {\lam \varx
      {\reset{\inctx {\ctx'_0}{\inctx \ctxzero \varx}}}}} \cloctc
  \empbisim \empbisim \reset{\subst \tm \vark {\lam \varx
      {\reset{\inctx {\ctx'_1}{\inctx \ctx \varx}}}}}$, as
  wished. The other case $\inctx \ctxzero {\shift \vark \tmzero} =
  \app {\inctx {\ctx''_0}{\shift \vark \tmzero}}{\tmzero'}$, $\inctx
  \ctx {\shift \vark \tm} = \app {\inctx {\ctx''}{\shift \vark
      \tm}}{\tm'}$ with $\tmzero' \cloctc\empbisim \tm'$ and $\inctx
       {\ctx''_0}{\shift \vark \tmzero} \cloctc\empbisim \inctx
       {\ctx''}{\shift \vark \tm}$ is treated similarly.\\

  We are now in a position to prove the lemma. We have just proved
  that $\reset{\subst \tmzero \vark {\lam \varx {\reset{\inctx
          {\ctx'_0}{\inctx \ctxzero \varx}}}}} \cloctc \empbisim
  \empbisim \reset{\subst \tm \vark {\lam \varx {\reset{\inctx
          {\ctx'_1}{\inctx \ctx \varx}}}}}$. We also have that
  $$\reset{\subst \tm \vark {\lam \varx {\reset{\inctx
          {\ctx'_1}{\inctx \ctx \varx}}}}} \empbisim \reset{\subst
    \tmone \vark {\lam \varx {\reset{\inctx {\ctx'_1}{\inctx \ctxone
            \varx}}}}}$$ by Lemma~\ref{l:shift}, therefore the required
  result holds by transitivity of $\empbisim$.
\end{proof}

\begin{lemma}
  \label{l:cong-empbisim}
  $\tmzero \empbisim \tmone$ implies $\inctx \cctx \tmzero \ierel
  \bisim {\rvals {\cloct \empbisim}} \inctx \cctx \tmone$.
\end{lemma}

\begin{proof}
  We prove that 
  $$\erel = \{ (\rvals {\cloct \empbisim}, \tmzero, \tmone) \mmid \tmzero
  \cloctc \empbisim \tmone\} \cup \{ \rvals {\cloct \empbisim} \}$$ is
  a bisimulation up-to bisimilarity. Let $\tmzero \ierel \erel {\rvals
    {\cloct \empbisim}} \tmone$. We prove clauses
  \ref{e:utbisim-tmtau}, \ref{e:utbisim-tmval}, and
  \ref{e:utbisim-tmstuck} of Definition \ref{d:utbisim} by induction
  on $\tmzero \cloctc \empbisim \tmone$. Note that by definition of
  $\erel$, we have $\tm \ierel \erel {\rvals {\cloct \empbisim}} \tm'$
  iff $\tm \cloctc \empbisim \tm'$.
    
  Suppose $\tmzero \empbisim \tmone$. This case holds because
  $\empbisim$ is an environmental bisimulation.

  Suppose $\tmzero = \lam \varx {\tmzero'}$, $\tmone = \lam \varx
  {\tmone'}$ with $\tmzero' \cloct \empbisim \tmone'$ and
  $\fv{\tmzero'} \cup \fv{\tmone'} \subseteq \{ \varx \}$. We have to
  prove that $(\rvals {\cloct \empbisim} \cup \{(\tmzero, \tmone)\})
  \in \erel$, i.e., $\rvals {\cloct \empbisim} \in \erel$, which is
  true.

  Suppose $\tmzero = \app{\tmzero^1}{\tmzero^2}$,
  $\tmone=\app{\tmone^1}{\tmzero^2}$ with $\tmzero^1 \cloctc \empbisim
  \tmone^1$ and $\tmzero^2 \cloctc \empbisim \tmone^2$. We distinguish
  several cases.
  \begin{itemize}
  \item If $\tmzero^1 \redcbv {\tmzero^1}'$, then $\tmzero \redcbv
    \app{{\tmzero^1}'}{\tmzero^2}$. By induction there exist
    ${\tmone^1}''$, ${\tmone^1}'$ such that $\tmone^1 \clocbv
        {\tmone^1}'$ and ${\tmzero^1}' \cloctc \empbisim {\tmone^1}''
        \empbisim {\tmone^1}'$. Consequently we have $\tmone \clocbv
        \app{{\tmone^1}'}{\tmone^2}$. By definition, we have
        $\app{{\tmzero^1}'}{\tmzero^2} \cloctc\empbisim
        \app{{\tmone^1}''}{\tmone^2}$, and by
        Lemma~\ref{l:cong-e-ctx}, we have
        $\app{{\tmone^1}''}{\tmone^2} \empbisim
        \app{{\tmone^1}'}{\tmone^2}$, hence
        $\app{{\tmzero^1}'}{\tmzero^2} \cloctc\empbisim \empbisim
        \app{{\tmone^1}'}{\tmone^2}$ holds, as wished.
  \item If $\tmzero^1 = \valzero$ and $\tmzero^2 \redcbv
    {\tmzero^2}'$, then $\tmzero \redcbv \app \valzero
    {{\tmzero^2}'}$. By induction there exist ${\tmone^2}''$,
    ${\tmone^2}'$ such that $\tmone^2 \clocbv {\tmone^2}'$ and
    ${\tmzero^2}' \cloctc \empbisim {\tmone^2}'' \empbisim
    {\tmone^2}'$. There also exists $\valone'$, $\valone$ such that
    $\tmone^1 \clocbv \valone$ and $\valzero \rvals {\cloct \empbisim}
    \valone' \empbisim \valone$. Consequently we have $\tmone \clocbv
    \app \valone {{\tmone^2}'}$. By definition, we have $\app \valzero
         {{\tmzero^2}'} \cloctc \empbisim
         \app{\valone'}{{\tmone^2}''}$, and by
         Lemma~\ref{l:cong-e-ctx} and transitivity of~$\empbisim$, we
         have $\app{\valone'}{{\tmone^2}''} \empbisim \app \valone
         {{\tmone^2}'}$, hence $\app \valzero {{\tmzero^2}'} \cloctc
         \empbisim \empbisim \app \valone {{\tmone^2}'}$ holds, as
         wished.
  \item If $\tmzero^1 = \lam \varx {\tmzero'}$ and $\tmzero^2 \redcbv
    \valzero$, then $\tmzero \redcbv \subst {\tmzero'} \varx
    \valzero$. By induction there exist $\tmone''$, $\tmone'$ such
    that $\tmone^1 \clocbv \lam \varx {\tmone'}$ and $\lam \varx
    {\tmzero'} \rvals {\cloct \empbisim} \lam \varx {\tmone''}
    \empbisim \lam \varx {\tmone'}$. There also exists $\valone'$,
    $\valone$ such that $\tmone^2 \clocbv \valone$ and $\valzero
    \rvals {\cloct \empbisim} \valone' \empbisim
    \valone$. Consequently we have $\tmone \clocbv \subst {\tmone'}
    \varx \valone$. By Lemma~\ref{l:subst-val}, we have $\subst
          {\tmzero'} \varx \valzero \cloctc\empbisim \empbisim \subst
          {\tmone'} \varx \valone$, as wished.
  \item If $\tmzero = \app{\inctx \ctxzero {\shift \vark
      {\tmzero'}}}{\tmzero^2}$, then by induction there exist
    $\ctxone$ and $\tmone'$ such that $\tmone^1 \clocbv \inctx \ctxone
    {\shift \vark {\tmone'}}$ and $\inctx \ctxzero {\shift \vark
      {\tmzero'}} \cloctc \empbisim \empbisim \inctx \ctxone {\shift
      \vark {\tmone'}}$. By definition of $\cloctc \empbisim$ and
    Lemma \ref{l:cong-e-ctx}, we have $\app{\inctx \ctxzero {\shift
        \vark {\tmzero'}}}{\tmzero^2} \cloctc \empbisim \empbisim
    \app{\inctx \ctxone {\shift \vark {\tmone'}}}{\tmone^2}$, hence
    the result holds. The reasoning is the same if $\tmzero = \app
    \valzero {\inctx \ctxzero {\shift \vark {\tmzero'}}}$.
  \end{itemize}
  Suppose $\tmzero = \reset{\tmzero'}$ and $\tmone = \reset{\tmone'}$
  with $\tmzero' \cloctc\empbisim \tmone'$. We have three cases to
  consider.
  \begin{itemize}
  \item If $\tmzero' \redcbv \tmzero''$, then $\tmzero \redcbv
    \reset{\tmzero''}$. By induction there exists $\tmone''$ such that
    $\tmone' \clocbv \tmone''$ and $\tmzero'' \cloctc\empbisim
    \empbisim \tmone''$. Consequently we have $\tmone \clocbv
    \reset{\tmone''}$, and by definition of $\cloctc\empbisim$ and
    Lemma \ref{l:cong-e-ctx}, we have $\reset{\tmzero''} \cloctc
    \empbisim \empbisim \reset{\tmone''}$.
  \item If $\tmzero'=\inctx \ctxzero {\shift \vark {\tmzero''}}$, then
    $\tmzero \redcbv \reset{\subst {\tmzero''} \vark {\lam \varx
      {\reset {\inctx \ctxzero \varx}}}}$. By induction, there exist
    $\ctxone$ and $\tmone''$ such that $\tmone' \clocbv \inctx \ctxone
    {\shift \vark {\tmone''}}$ and $\inctx \ctxzero {\shift \vark
      {\tmzero''}} \cloctc \empbisim \empbisim \inctx \ctxone {\shift
      \vark {\tmone''}}$. By Lemma \ref{l:subst-shift}, we have
    $\reset{\subst {\tmzero''} \vark {\lam \varx {\reset {\inctx
            \ctxzero \varx}}}} \cloctc \empbisim \empbisim
    \reset{\subst {\tmone''} \vark {\lam \varx {\reset {\inctx \ctxone
            \varx}}}}$, as wished.
  \item If $\tmzero' = \valzero$, then $\tmzero \redcbv \valzero$. By
    induction, there exists $\valone$ such that $\tmone' \clocbv
    \valone$ and $\valzero \cloctc\empbisim \empbisim \valone$. We
    have $\tmone \clocbv \valone$, hence the result holds.
  \end{itemize}
  Suppose $\tmzero = \shift \vark{\tmzero'}$ and $\tmone = \shift
  \vark {\tmone'}$ with $\tmzero' \cloct \empbisim \tmone'$ and
  $\fv{\tmzero'} \cup \fv{\tmone'} \subseteq \{ \varx \}$. We have to
  prove that $(\rvals {\cloct \empbisim} \cup \{(\tmzero, \tmone)\})
  \in \erel$, i.e., $\rvals {\cloct \empbisim} \in \erel$, which is
  true.\\

  We now prove items \ref{e:utbisim-envv} and \ref{e:utbisim-envs} of
  Definition~\ref{d:utbisim}. Suppose $\lam \varx \tmzero \cloctc
  \empbisim \lam \varx \tmone$ and $\valzero \cloctc \empbisim
  \valone$. Then by Lemma \ref{l:subst-val} and reflexivity of
  $\empbisim$, we have $\subst \tmzero \varx \valzero \empbisim
  \cloctc \empbisim \empbisim \subst \tmone \varx \valzero$, as
  wished.

  Suppose $\inctx \ctxzero {\shift \vark \tmzero} \cloctc \empbisim
  \inctx \ctxone {\shift \vark \tmone}$ and $\ctx'_0 \clocc
  \empbisim \ctx'_1$. Then by Lemma \ref{l:subst-shift} and
  reflexivity of $\empbisim$, $\reset {\subst \tmzero \vark {\lam
      \varx {\reset {\inctx {\ctx'_0}{\inctx \ctxzero \varx}}}}}
  \empbisim \cloctc \empbisim \empbisim \reset {\subst \tmone \vark
    {\lam \varx {\reset {\inctx {\ctx'_1}{\inctx \ctxone \varx}}}}}$,
  as wished.
\end{proof}

\begin{remark}
  The proof of Lemma~\ref{l:cong-empbisim} uses up to bisimilarity
  because of Lemma~\ref{l:subst-val}, and Lemma~\ref{l:subst-val}
  cannot be strengthened.
\end{remark}

\begin{corollary}
  For all $\env \in \mathord{\bisim}$, $\ebisim$ is a congruence.
\end{corollary}

\begin{proof}
  If $\tmzero \ebisim \tmone$, then by weakening
  (Lemma~\ref{l:weakening}), we have $\tmzero \empbisim \tmone$, which
  in turn implies $\inctx \cctx \tmzero \ierel \bisim {\rvals {\cloct
      \empbisim}} \inctx \cctx \tmone$ (by
  Lemma~\ref{l:cong-empbisim}), and gives us $\inctx \cctx \tmzero
  \ierel \bisim \env \inctx \cctx \tmone$ using weakening again.
\end{proof}

\begin{corollary}[Soundness]
  \label{c:soundness}
  The relation $\empbisim$ is sound.
\end{corollary}

\begin{proof}
  Because it is a congruence, and the observables actions coincide.
\end{proof}

\begin{theorem}[Completeness]
  The relation $\empbisim $ is complete.
\end{theorem}

\begin{proof}
  We prove that $\erel = \{(\rvals \ctxequiv, \tmzero, \tmone) \mmid
  \tmzero \ctxequiv \tmone\} \cup \{ \rvals \ctxequiv \}$ is a
  big-step environmental bisimulation.

  Let $\tmzero \ierel \erel {\rvals \ctxequiv} \tmone$. If $\tmzero
  \clocbv \valzero$, then by definition of $\ctxequiv$, there exists
  $\valone$ such that $\tmone \clocbv \valone$. By
  Lemma~\ref{l:redcbv-in-empbisim}, we have $\tmzero \empbisim
  \valzero$ and $\tmone \empbisim \valone$, which implies $\tmzero
  \ctxequiv \valzero$ and $\tmone \ctxequiv \valone$ by
  Corollary~\ref{c:soundness}. Transitivity of $\ctxequiv$ gives
  $\valzero \ctxequiv \valone$, hence we have $\rvals \ctxequiv \cup
  \{(\valzero, \valone)\} = \mathord{\rvals \ctxequiv} \in \erel$, as
  wished. The reasoning is the same for $\tmzero \clocbv \tmzero'$,
  where $\tmzero'$ is a stuck term.

  Let $\lam \varx \tmzero \ctxequiv \lam \varx \tmone$ and $\valzero
  \cloctc{\rvals \ctxequiv} \valone$. By congruence of $\ctxequiv$, we
  have $\valzero \ctxequiv \valone$, and also $\app {\lamp \varx
    \tmzero} \valzero \ctxequiv \app {\lamp \varx \tmone}
  \valone$. Because $\app {\lamp \varx \tmzero} \redcbv \subst \tmzero
  \varx \valzero$, $\app {\lamp \varx \tmone} \valone \redcbv \subst
  \tmone \varx \valone$ and $\redcbv \subseteq \mathord{\empbisim}
  \subseteq \mathord{\ctxequiv}$, we have $\subst \tmzero \varx
  \valzero \ctxequiv \subst \tmone \varx \valone$, i.e., $\subst
  \tmzero \varx \valzero \ierel \erel {\rvals \ctxequiv} \subst \tmone
  \varx \valone$, as wished.

  Let $\inctx \ctxzero {\shift \vark \tmzero} \ctxequiv \inctx \ctxone
  {\shift \vark \tmone}$ and $\ctx'_0 \clocc {\rvals \ctxequiv}
  \ctx'_1$. By induction on $\ctx'_0 \clocc {\rvals
    \ctxequiv} \ctx'_1$, we know that $\reset{\inctx {\ctx'_0}{\inctx
      \ctxzero {\shift \vark \tmzero}}} \ctxequiv \reset{\inctx
    {\ctx'_1}{\inctx \ctxone {\shift \vark \tmone}}}$ holds. From
   $\reset{\inctx {\ctx'_0}{\inctx \ctxzero {\shift \vark
        \tmzero}}} \redcbv \reset{\subst \tmzero \vark {\lam \varx
      \reset{\inctx {\ctx'_0}{\inctx \ctxzero \varx}}}}$, 
  $\reset{\inctx {\ctx'_1}{\inctx \ctxone {\shift \vark \tmone}}}
  \redcbv \reset{\subst \tmone \vark {\lam \varx \reset{\inctx
        {\ctx'_1}{\inctx \ctxone \varx}}}}$, and $\redcbv \subseteq
  \mathord{\empbisim} \subseteq \mathord{\ctxequiv}$, we get
  $\reset{\subst \tmzero \vark {\lam \varx \reset{\inctx
        {\ctx'_0}{\inctx \ctxzero \varx}}}} \ctxequiv \reset{\subst
    \tmone \vark {\lam \varx \reset{\inctx {\ctx'_1}{\inctx \ctxone
          \varx}}}}$, and then, as required, $\reset{\subst \tmzero
    \vark {\lam \varx \reset{\inctx {\ctx'_0}{\inctx \ctxzero
          \varx}}}} \ierel \erel {\rvals \ctxequiv} \reset{\subst
    \tmone \vark {\lam \varx \reset{\inctx {\ctx'_1}{\inctx \ctxone
          \varx}}}}$.
\end{proof} 

\section{Soundness and Completeness for the Original Semantics}
\label{a:programs}

\begin{lemma}
  \label{l:weakening-p}
  If $\tmzero \ebisimp \tmone$ and $\env' \subseteq \env$ then
  $\tmzero \ierel \bisimp {\env'} \tmone$.
\end{lemma}

\begin{proof}
  As in \cite{Sangiorgi-al:TOPLAS11}.
\end{proof}

\begin{lemma}
  \label{l:redcbv-in-empbisim-p}
  If $\tmzero \redcbv \tmzero'$, then $\tmzero \empbisimp \tmzero'$.
\end{lemma}

\begin{proof}
  Same as for Lemma~\ref{l:redcbv-in-empbisim}.
\end{proof}

\begin{lemma}
  \label{l:cong-val-p}
  For all $\env$, if $\valzero \ebisim \valone$, then $\inctx \cctx
  \valzero \ierel \bisimp {\rval{\cloct \env}} \inctx \cctx \valone$.
\end{lemma}

\begin{lemma}
  \label{l:cong-e-ctx-p}
  For all $\env$, if $\tmzero \ebisim \tmone$, then $\inctx \rctx
  \tmzero \ierel \bisimp {\rval{\cloct \env}} \inctx \rctx \tmone$.
\end{lemma}
We prove Lemmas~\ref{l:cong-val-p} and~\ref{l:cong-e-ctx-p}
simultaneously. Let $\erely$ be an environmental bisimulation. We
define
\begin{align*}
  \erel & = \erel_1 \cup \erel_2 \cup \{\rval{\cloct \env} \mmid \env \in \mathord{\erely} \}
  \\
  \erel_1 & = \{(\rval {\cloct \env}, \inctx \rctxzero \tmzero, \inctx \rctxone
  \tmone) \mmid \tmzero \ierel \erely \env \tmone, \rctxzero \clocc \env \rctxone\}
  \\
  \erel_2 & = \{(\rval {\cloct \env}, \tmzero, \tmone) \mmid \env \in \mathord{\erely},
  \tmzero \cloctc \env \tmone \}
\end{align*}
In $\erel_2$, we build the closed terms $(\tmzero, \tmone)$ out of
pairs of values. We first prove a preliminary lemma about
$\erel$. Remark that $\erel$ is a congruence.

\begin{lemma}
  \label{l:erel-p}
  Let $\env \in \mathord{\erely}$. If $\lam \varx \tmzero \cloctc \env
  \lam \varx \tmone$ and $\valzero \rval{\cloct \env} \valone$ then
  $\subst \tmzero \varx \valzero \ierel \erel {\rval {\cloct \env}}
  \subst \tmone \varx \valone$.
\end{lemma}

\begin{proof}
  We proceed by case analysis on $\lam \varx \tmzero \cloctc \env \lam
  \varx \tmone$. If $\lam \varx \tmzero \env \lam \varx \tmone$, then
  because $\erely$ is an environmental bisimulation, we have $\subst
  \tmzero \varx \valzero \ierel \erely \env \subst \tmone \varx
  \valone$, which implies $\subst \tmzero \varx \valzero \ierel \erel
         {\rval {\cloct \env}} \subst \tmone \varx \valone$ (more
         precisely, the terms are in $\erel_1$).

  If $\tmzero \cloct \env \tmone$ with $\fv\tmzero \cup \fv\tmone
  \subseteq \{\varx\}$, then we have $\subst \tmzero \varx \valzero
  \cloct \env \subst \tmone \varx \valone$ by
  Lemma~\ref{l:cong-is-subst}. In fact, we have $\subst \tmzero \varx
  \valzero \cloctc \env \subst \tmone \varx \valone$, so we have
  $\subst \tmzero \varx \valzero \ierel \erel {\rval {\cloct \env}}
  \subst \tmone \varx \valone$ (more precisely, the terms are in
  $\erel_2$).
\end{proof}

We now prove Lemmas~\ref{l:cong-val-p} and~\ref{l:cong-e-ctx-p} by
showing that $\erel$ is a bisimulation up to environment.
\begin{proof}
  We first prove the bisimulation for the elements in $\erel_2$ (for
  these, we do not need the ``up to environment''). Let $\tmzero
  \cloctc \env \tmone$, with $\env \in \mathord{\erely}$. If $\tmzero$
  or $\tmone$ is not a program, then for all $\ctxzero \clocc \env
  \ctxone$, we have $\reset{\inctx \ctxzero \tmzero} \cloctc \env
  \reset{\inctx \ctxone \tmone}$, i.e., $\reset{\inctx \ctxzero
    \tmzero} \ierel \erel {\rval{\cloct\env}} \reset{\inctx \ctxone
    \tmone}$, hence clause \ref{e:tm-p} holds. Suppose $\tmzero$,
  $\tmone$ are programs $\prgzero$, $\prgone$. If $\prgzero \redcbv
  \valzero$, then $\prgzero = \reset \valzero$, and therefore $\prgone
  = \reset \valone$ with $\valzero \cloctc \erel \valone$. We have
  $\prgone \redcbv \valone$, and also $\{(\valzero,\valone)\} \cup
  \rval{\cloct \env} = \mathord{\rval{\cloct \env}} \in \erel$, as
  wished by clause \ref{e:prgval-p}.

  Otherwise $\prgzero \redcbv \prgzero'$. Then $\prgzero = \inctx
  \rctxzero \redexzero$. Because $\env$ relates only values, we can
  prove there exist $\rctxone$, $\redexone$ such that $\prgone =
  \inctx \rctxone \redexone$, $\rctxzero \clocc\env \rctxone$, and
  $\redexzero \cloctc\env \redexone$. We show that clause
  \ref{e:prgtau-p} holds by case analysis on the different redexes.
  \begin{itemize}
  \item If $\redexzero = \reset \valzero$ and $\redexone = \reset
    \valone$ with $\valzero \rval{\cloct\env} \valone$, then $\prgzero
    \redcbv \inctx \rctxzero \valzero$ and $\prgone \redcbv \inctx
    \rctxone \valone$. We have $\inctx \rctxzero \valzero \cloct\env
    \inctx \rctxone \valone$, as wished.
  \item Suppose $\redexzero = \app{\lamp \varx {\tmzero'}} \valzero$
    and $\redexone = \app{\lamp \varx {\tmone'}} \valone$ with $\lam
    \varx {\tmzero'} \rval{\cloct\env} \lam \varx {\tmone'}$ and
    $\valzero \rval{\cloct\env} \valone$. Then $\prgzero \redcbv
    \inctx \rctxzero {\subst {\tmzero'} \varx \valzero}$ and $\prgone
    \redcbv \inctx \rctxone {\subst {\tmone'} \varx \valone}$. By
    Lemma~\ref{l:erel-p} and because $\erel$ is a congruence, we have
    $\inctx \rctxzero {\subst {\tmzero'} \varx \valzero} \ierel \erel
    {\rval{\cloct \env }} \inctx \rctxone {\subst {\tmone'} \varx
      \valone}$, as wished.
  \item If $\redexzero = \reset{\inctx \ctxzero {\shift \vark
      {\tmzero'}}}$ and $\redexone = \reset{\inctx \ctxone {\shift
      \vark {\tmone'}}}$ with $\ctxzero \clocc\env \ctxone$ and
    $\tmzero' \rval{\cloct\env} \tmone'$. Then $\prgzero \redcbv
    \inctx \rctxzero {\reset{\subst {\tmzero'} \vark {\lam \varx
          \reset{\inctx \ctxzero \varx}}}}$ and $\prgone \redcbv
    \inctx \rctxone {\reset{\subst {\tmone'} \vark {\lam \varx
          \reset{\inctx \ctxone \varx}}}}$. From $\ctxzero \clocc\env
    \ctxone$, we deduce $\lam \varx \reset{\inctx \ctxzero \varx}
    \rval{\cloct\env} \lam \varx \reset{\inctx \ctxone \varx}$, so by
    Lemma~\ref{l:cong-is-subst}, we have $\inctx \rctxzero
    {\reset{\subst{\tmzero'} \vark {\lam \varx \reset{\inctx \ctxzero
            \varx}}}} \cloct \env \inctx \rctxone
    {\reset{\subst{\tmone'} \vark {\lam \varx \reset{\inctx \ctxone
            \varx}}}}$, as wished.
  \end{itemize}

  We now prove the bisimulation property (up to environment) for
  elements in $\erel_1$. Let $\inctx \rctxzero \tmzero \ierel \erel
  {\rval{\cloct \env}} \inctx \rctxone \tmone$, so that $\tmzero
  \ierel \erely \env \tmone$ and $\rctxzero \clocc \env \rctxone$. If
  $\inctx \rctxzero \tmzero$ and $\inctx \rctxone \tmone$ are not both
  programs, then for all $\ctxzero \clocc \env \ctxone$, we have
  $\reset{\inctx \ctxzero {\inctx \rctxzero \tmzero}} \ierel \erel
  {\rval{\cloct\env}} \reset{\inctx \ctxone {\inctx \rctxone
      \tmone}}$, hence clause \ref{e:tm-p} holds. Suppose $\inctx
  \rctxzero \tmzero$, $\inctx \rctxone \tmone$ are programs
  $\prgzero$, $\prgone$. We distinguish two cases. First, suppose
  $\tmzero$ and $\tmone$ are programs. If $\tmzero \redcbv \prgzero'$,
  then $\prgzero \redcbv \inctx \rctxzero {\prgzero'}$. Because
  $\tmzero \ierel \erely \env \tmone$, there exists $\prgone'$ such
  that $\tmone \clocbv \prgone'$ and $\prgzero' \ierel \erely \env
  \prgone'$. We have $\inctx \rctxzero {\prgzero'} \ierel \erel
          {\rval{\cloct \env}} \inctx \rctxone {\prgone'}$ and
          $\prgone \clocbv \inctx \rctxone{\prgone'}$, therefore
          clause \ref{e:prgtau-p} holds. If $\tmzero \redcbv
          \valzero$, then $\prgzero \redcbv \inctx \rctxzero
          \valzero$. Because $\tmzero \ierel \erely \env \tmone$,
          there exists $\valone$ such that $\tmone \clocbv \valone$
          and $\env' = \{(\valzero, \valone)\} \cup \env \in
          \mathord{\erely}$. Hence $\prgone \clocbv \inctx \rctxone
          \valone$, and we have $\inctx \rctxzero \valzero \ierel
          \erel {\rval{\cloct {\env'}}} \inctx \rctxone \valone$,
          therefore clause \ref{e:prgtau-p} holds (up to environment).

  In the second case, $\tmzero$ and $\tmone$ are not both
  programs. Then we can write $\prgzero =
  \inctx{\rctxzero'}{\reset{\inctx \ctxzero \tmzero}}$ and $\prgone =
  \inctx{\rctxone'}{\reset{\inctx \ctxone \tmone}}$ for some
  $\rctxzero' \clocc\env \rctxone'$ and $\ctxzero \clocc\env
  \ctxone$. Because $\tmzero \ierel \erely \env \tmone$ and since
  $\erely$ is an environmental bisimulation, we have $\reset{\inctx
    \ctxzero \tmzero} \ierel \erely \env \reset{\inctx \ctxone
    \tmone}$. If $\reset{\inctx \ctxzero \tmzero} \redcbv \prgzero'$,
  then there exists $\prgone'$ such that $\reset{\inctx \ctxone
    \tmone} \clocbv \prgone'$ and $\prgzero' \ierel \erely \env
  \prgone'$. Therefore, $\prgzero \redcbv \inctx
          {\rctxzero'}{\prgzero'}$, $\prgone \redcbv \inctx
          {\rctxone'}{\prgone'}$, and $\inctx {\rctxzero'}{\prgzero'}
          \ierel \erel {\rval{\cloct \env}} \inctx
                 {\rctxone'}{\prgone'}$, hence clause \ref{e:prgtau-p}
                 holds. If $\reset{\inctx \ctxzero \tmzero} \redcbv
                 \valzero$, then there exists $\valone$ such that
                 $\reset{\inctx \ctxone \tmone} \clocbv \valone$ and
                 $\env' = \{(\valzero, \valone)\} \cup \env \subseteq
                 \mathord{\erely}$. Therefore $\prgzero \redcbv \inctx
                         {\rctxzero'} \valzero$, $\prgone \redcbv
                         \inctx {\rctxone'}{\valone}$, and $\inctx
                                {\rctxzero'}{\valzero} \ierel \erel
                                {\rval{\cloct {\env'}}} \inctx
                                {\rctxone'}{\valone}$, hence clause
                                \ref{e:prgtau-p} holds (up to
                                environment).\\

  We finally prove the clause~\ref{e:env-p} of the bisimulation. The
  only environments in $\erel$ are of the form $\rval {\cloct
    \env}$. Let $\lam \varx \tmzero \rval {\cloct \env} \lam \varx
  \tmone$ and $\valzero \rval {\cloct \env} \valone$. By
  Lemma~\ref{l:erel-p}, we have $\subst \tmzero \varx \valzero \ierel
  \erel {\rval {\cloct \env}} \subst \tmone \varx \valone$, hence the
  result holds.
\end{proof}

\begin{lemma}
  \label{l:val-p}
  If $\lam \varx \tmzero \empbisimp \lam \varx \tmone$, then $\subst
  \tmzero \varx \val \empbisimp \subst \tmone \varx \val$.
\end{lemma}

\begin{proof}
  If $\lam \varx \tmzero \empbisimp \lam \varx \tmone$, then
  $\reset{\lam \varx \tmzero} \empbisimp \reset{\lam \varx \tmone}$ by
  clause \ref{e:tm-p}. Since $\reset{\lam \varx \tmzero} \redcbv \lam
  \varx \tmzero$ and $\reset{\lam \varx \tmone} \redcbv \lam \var
  \tmone$, we have $\{(\lam \varx \tmzero, \lam \varx \tmone) \} \in
  \mathord{\empbisimp}$ by clause \ref{e:prgval-p}. Let $\env =
  \{(\lam \varx \tmzero, \lam \varx \tmone) \}$. By clause
  \ref{e:env-p}, for all $\val$, we have $\subst \tmzero \varx \val
  \ebisim \subst \tmone \varx \val$, therefore $\subst \tmzero \varx
  \val \empbisimp \subst \tmone \varx \val$ holds by weakening
  (Lemma~\ref{l:weakening}).
\end{proof}

\begin{lemma}
  \label{l:subst-val-p}
  If $\lam \varx \tmzero \cloctc \empbisimp \lam \varx \tm \empbisimp
  \lam \varx \tmone$ and $\valzero \cloctc \empbisimp \val \empbisimp
  \valone$ then $\subst \tmzero \varx \valzero \cloctc \empbisimp
  \empbisimp \subst \tmone \varx \valone$.
\end{lemma}

\begin{proof}
  We proceed by case analysis on $\lam \varx \tmzero \cloctc
  \empbisimp \lam \varx \tm$.

  Suppose $\lam \varx \tmzero \empbisimp \lam \varx \tm$. We have
  $\subst \tmzero \varx \valzero \cloctc \empbisimp \subst \tmzero
  \varx \val$ by Lemma \ref{l:cong-is-subst}, $\subst \tmzero \varx
  \val \empbisimp \subst \tm \varx \val$ by Lemma~\ref{l:val-p},
  $\subst \tm \varx \val \empbisimp \subst \tm \varx \valone$ by
  Lemma~\ref{l:cong-val-p}, and $\subst \tm \varx \valone \empbisimp
  \subst \tmone \varx \valone$ by Lemma~\ref{l:val-p}. Finally,
  $\subst \tmzero \varx \valzero \cloctc \empbisimp \empbisimp \subst
  \tmone \varx \valone$ holds using transitivity of $\empbisimp$.

  Suppose $\tmzero \cloct \empbisimp \tm$ with $\fv \tmzero \cup \fv
  \tm \subseteq \{ \varx \}$. We have $\subst \tmzero \varx \valzero
  \cloctc\empbisimp \subst \tm \varx \val$ by
  Lemma~\ref{l:cong-is-subst}, $\subst \tm \varx \val \empbisimp
  \subst \tm \varx \valone$ by Lemma~\ref{l:cong-val-p}, and $\subst
  \tm \varx \valone \empbisimp \subst \tmone \varx \valone$ by
  Lemma~\ref{l:val-p}. Finally, $\subst \tmzero \varx \valzero \cloctc
  \empbisimp \empbisimp \subst \tmone \varx \valone$ holds using
  transitivity of $\empbisimp$.
\end{proof}

Programs are either value programs or can be decomposed in contexts
$\rctx$, $\ctx$, and a redex $\redex$. We extend this result to
related programs $\prgzero \cloctc\empbisimp \prgone$, and see how
they can be decomposed.
\begin{lemma}
  \label{l:decompose-prg}
  If $\prgzero \cloctc\empbisimp \prgone$ then we have one of the
  following cases:
  \begin{itemize}
  \item $\prgzero \empbisimp \prgone$;
  \item $\prgzero = \reset \valzero$;
  \item $\prgzero = \inctx \rctxzero {\reset{\inctx \ctxzero
      \tmzero}}$, $\prgone = \inctx \rctxone {\reset{\inctx \ctxone
      \tmone}}$ , $\rctxzero \clocc \empbisimp \rctxone$, $\ctxzero
    \clocc\empbisimp \ctxone$, $\tmzero \empbisimp \tmone$ and
    $\tmzero \redcbv \tmzero'$ or $\tmzero$ is stuck;
  \item $\prgzero = \inctx \rctxzero {\reset{\inctx \ctxzero
      \redexzero}}$, $\prgone = \inctx \rctxone {\reset{\inctx \ctxone
      \tmone}}$ , $\rctxzero \clocc \empbisimp \rctxone$, $\ctxzero
    \clocc\empbisimp \ctxone$, $\redexzero \cloctc\empbisimp \tmone$
    but $\redexzero \not\empbisimp \tmone$.
  \end{itemize}
\end{lemma}

\begin{proof}
  We prove a more general result on $\tmzero \cloctc\empbisimp
  \tmone$. We have either
  \begin{itemize}
  \item $\tmzero \empbisimp \tmone$;
  \item $\tmzero = \valzero$;
  \item $\tmzero = \inctx \ctxzero {\tmzero'}$, $\tmone = \inctx
    \ctxone {\tmone'}$, $\ctxzero \clocc\empbisimp \ctxone$, $\tmzero'
    \empbisimp \tmone'$, and $\tmzero' \redcbv \tmzero''$ or $\tmzero$
    is stuck;
  \item $\tmzero = \inctx \rctxzero {\reset{\inctx \ctxzero
      {\tmzero'}}}$, $\tmone = \inctx \rctxone {\reset{\inctx \ctxone
      {\tmone'}}}$ , $\rctxzero \clocc \empbisimp \rctxone$, $\ctxzero
    \clocc\empbisimp \ctxone$, $\tmzero' \empbisimp \tmone'$, and
    $\tmzero' \redcbv \tmzero''$ or $\tmzero'$ is stuck;
  \item $\tmzero = \inctx \ctxzero \redexzero$, $\tmone = \inctx
    \ctxone {\tmone'}$, $\ctxzero \clocc\empbisimp \ctxone$,
    $\redexzero \cloctc\empbisimp \tmone'$ but $\redexzero
    \not\empbisimp \tmone'$.
  \item $\tmzero = \inctx \rctxzero {\reset{\inctx \ctxzero
      \redexzero}}$, $\tmone = \inctx \rctxone {\reset{\inctx \ctxone
      {\tmone'}}}$ , $\rctxzero \clocc \empbisimp \rctxone$, $\ctxzero
    \clocc\empbisimp \ctxone$, $\redexzero \cloctc\empbisimp \tmone'$
    but $\redexzero \not\empbisimp \tmone'$.
  \end{itemize}
  The proof is easy by induction on $\tmzero \cloctc\empbisimp \tmone$
  but tedious.
\end{proof}

\begin{lemma}
  \label{l:bisim-val-p}
  If $\valzero \cloctc \empbisimp \tmone$, then there exists $\valone$
  such that $\reset \tmone \clocbv \valone$ and $\valzero \cloctc
  \empbisimp \valone$.
\end{lemma}

\begin{proof}
  We have two cases to consider. If $\valzero \empbisimp \tmone$, then
  $\reset \valzero \empbisimp \reset \valone$, and because $\reset
  \valzero \redcbv \valzero$, there exists $\valone$ such that
  $\reset\tmone \clocbv \valone$ and $\{(\valzero, \valone)\} \in
  \bisimp$. Because $\clocbv \subseteq \empbisimp$, we have $\valzero
  \empbisimp \valone$, as wished. Otherwise, $\tmone$ is a value
  $\valone$, and the result holds trivially.
\end{proof}

\begin{lemma}
  \label{l:progress-bisim-p}
  Let $\tmzero \empbisimp \tmone$ so that $\tmzero \redcbv \tmzero'$
  or $\tmzero$ is stuck, and $\ctxzero \clocc \empbisimp
  \ctxone$. There exist $\prgzero'$, $\prgone'$ such that
  $\reset{\inctx \ctxzero \tmzero} \redcbv \prgzero'$, $\reset{\inctx
    \ctxone \tmone} \clocbv \prgone'$, and $\prgzero' \cloctc
  \empbisimp \empbisimp \prgone'$.
\end{lemma}

\begin{proof}
  Suppose $\tmzero$ and $\tmone$ are both programs. Then $\tmzero$
  cannot be stuck, and we have $\tmzero \redcbv \tmzero'$. By
  bisimilarity, there exists $\prgone'$ such that $\tmone \clocbv
  \prgone'$ and $\tmzero' \empbisimp \prgone'$. We have $\reset{\inctx
    \ctxzero \tmzero} \redcbv \reset{\inctx \ctxzero {\tmzero'}}$,
  $\reset{\inctx \ctxone \tmone} \clocbv \reset{\inctx
    \ctxone{\prgone'}}$, and $\reset{\inctx \ctxzero {\tmzero'}}
  \cloctc \empbisimp \reset{\inctx \ctxone{\prgone'}}$, hence the
  result holds.

  Suppose $\tmzero$ and $\tmone$ are not both programs. Because
  $\tmzero \empbisimp \tmone$, we have $\reset{\inctx \ctxone \tmzero}
  \empbisimp \reset{\inctx \ctxone \tmone}$. From $\tmzero \redcbv
  \tmzero'$ or $\tmzero$ is stuck, we know that $\reset{\inctx \ctxone
    \tmzero} \redcbv \prgzero''$ for some $\prgzero''$ and
  $\reset{\inctx \ctxzero \tmzero} \redcbv \prgzero'$ with $\prgzero'
  \cloctc \empbisimp \prgzero''$. By bisimilarity, there exists
  $\prgone'$ such that $\reset{\inctx \ctxone \tmone} \clocbv
  \prgone'$ and $\prgzero'' \empbisimp \prgone'$. We have $\prgzero'
  \cloctc \empbisimp \empbisimp \prgone'$, hence the result holds.
\end{proof}

\begin{lemma}
  \label{l:bisim-beta-redex-p}
  Let $\lam \varx \tmzero \cloctc\empbisimp \tmone^1$, $\valzero
  \cloctc\empbisimp \tmone^2$, and $\ctxzero \clocc \empbisimp
  \ctxone$. There exist $\prgzero$, $\prgone$ such that $\reset{\inctx
    \ctxzero {\app {\lamp \varx \tmzero} \valzero}} \redcbv \prgzero$,
  $\reset{\inctx \ctxone {\app{\tmone^1}{\tmone^2}}} \clocbv \prgone$,
  and $\prgzero \cloctc\empbisimp \empbisimp \prgone$.
\end{lemma}

\begin{proof}
  We have four different cases:
  \begin{itemize}
  \item Suppose $\lam \varx \tmzero \cloctc\empbisimp \lam \varx
    \tmone$, $\valzero \cloctc\empbisimp \valone$. By Lemma
    \ref{l:subst-val-p}, we have $\subst \tmzero \varx \valzero
    \cloctc\empbisimp \empbisimp \subst \tmone \varx \valone$. We have
    $\reset{\inctx \ctxzero {\app {\lamp \varx \tmzero} \valzero}}
    \redcbv \reset{\inctx \ctxzero {\subst \tmzero \varx \tmzero}}$,
    $\reset{\inctx \ctxone {\app {\lamp \varx \tmone} \valone}}
    \redcbv \reset{\inctx \ctxone {\subst \tmone \varx \tmone}}$, and
    we have $\reset{\inctx \ctxzero {\subst \tmzero \varx \tmzero}}
    \cloctc\empbisimp \empbisimp \reset{\inctx \ctxone {\subst \tmone
        \varx \tmone}}$ by congruence of $\cloctc\empbisimp$ and Lemma
    \ref{l:cong-e-ctx-p}. Therefore, we have the required result.
  \item Suppose $\lam \varx \tmzero \cloctc\empbisimp \lam \varx
    \tmone$ with $\tmzero \cloctc \empbisimp \tmone$ and $\valzero
    \empbisimp \tmone^2$. By bisimilarity, we have $\reset{\inctx
      \ctxone {\app {\lamp \varx \tmone} \valzero}} \empbisimp
    \reset{\inctx \ctxone {\app {\lamp \varx \tmone}{\tmone^2}}}$,
    therefore there exists $\prgone$ such that $\reset{\inctx \ctxone
      {\app {\lamp \varx \tmone}{\tmone^2}}} \clocbv \prgone$ and
    $\reset{\inctx \ctxone {\subst \tmone \varx \valzero}} \empbisimp
    \prgone$. We get $\reset{\inctx \ctxzero {\app {\lamp \varx
          \tmzero} \valzero}} \redcbv \reset{\inctx \ctxzero {\subst
        \tmzero \varx \valzero}}$ and $\reset{\inctx \ctxzero {\subst
        \tmzero \varx \valzero}} \cloctc\empbisimp \reset{\inctx
      \ctxone {\subst \tmone \varx \valzero}}$, and hence the result
    holds.
  \item Suppose $\lam \varx \tmzero \empbisimp \tmone^1$ and $\valzero
    \cloctc\empbisimp \valone$. By bisimilarity, we have
    $\reset{\inctx \ctxone {\app {\lamp \varx \tmzero} \valzero}}
    \empbisimp \reset{\inctx \ctxone {\app {\tmone^1} \valone}}$,
    therefore there exists a program $\prgone$ such that
    $\reset{\inctx \ctxone {\app {\tmone^1} \valone}} \clocbv \prgone$
    and $\reset{\inctx \ctxone {\subst \tmzero \varx \valzero}}
    \empbisimp \prgone$. We have $\reset{\inctx \ctxzero {\app {\lamp
          \varx \tmzero} \valzero}} \redcbv \reset{\inctx \ctxzero
      {\subst \tmzero \varx \valzero}}$ and $\reset{\inctx \ctxzero
      {\subst \tmzero \varx \valzero}} \cloctc\empbisimp \reset{\inctx
      \ctxone {\subst \tmzero \varx \valzero}}$, hence the result
    holds.
  \item Suppose $\lam \varx \tmzero \empbisimp \tmone^1$ and $\valzero
    \empbisimp \tmone^2$. By bisimilarity, we have $\reset{\inctx
      \ctxone {\app {\lamp \varx \tmzero} \valzero}} \empbisimp
    \reset{\inctx \ctxone {\app {\tmone^1} \valzero}}$ and
    $\reset{\inctx \ctxone {\app {\tmone^1} \valzero}} \empbisimp
    \reset{\inctx \ctxone {\app {\tmone^1}{\tmone^2}}}$, therefore
    there exists $\prgone$ such that $\reset{\inctx \ctxone {\app
        {\tmone^1}{\tmone^2}}} \clocbv \prgone$ and $\reset{\inctx
      \ctxone {\subst \tmzero \varx \valzero}} \empbisimp \prgone$. We
    get $\reset{\inctx \ctxzero {\app {\lamp \varx \tmzero} \valzero}}
    \redcbv \reset{\inctx \ctxzero {\subst \tmzero \varx \valzero}}$
    and $\reset{\inctx \ctxzero {\subst \tmzero \varx \valzero}}
    \cloctc\empbisimp \reset{\inctx \ctxone {\subst \tmzero \varx
        \valzero}}$, and hence the result holds.
  \end{itemize}
\end{proof}

\begin{lemma}
  \label{l:bisim-stuck-p}
  If $\inctx \ctx {\shift \vark \tmzero} \cloctc \empbisimp \tmone$
  and $\ctxzero \clocc \empbisimp \ctxone$, then there exist
  $\prgzero$, $\prgone$ such that $\reset{\inctx \ctxzero {\inctx \ctx
      {\shift \vark \tmzero}}} \redcbv \prgzero$, $\reset{\inctx
    \ctxone \tmone} \clocbv \prgone$, and $\prgzero \cloctc \empbisimp
  \empbisimp \prgone$.
\end{lemma}

\begin{proof}
  We proceed by induction on $\inctx \ctx {\shift \vark \tmzero}
  \cloctc \empbisimp \tmone$.

  If $\inctx \ctx {\shift \vark \tmzero} \empbisimp \tmone$, then by
  bisimilarity we have $\reset{\inctx \ctxone {\inctx \ctx {\shift
        \vark \tmzero}}} \empbisimp \reset{\inctx \ctxone
    \tmone}$. Because $\reset{\inctx \ctxone {\inctx \ctx {\shift
        \vark \tmzero}}} \redcbv \prgzero'$, there exists $\prgone$
  such that $\reset{\inctx \ctxone \tmone} \clocbv \prgone$ and
  $\prgzero' \empbisimp \prgone$. We also have $\reset{\inctx \ctxzero
    {\inctx \ctx {\shift \vark \tmzero}}} \redcbv \prgzero$ with
  $\prgzero \cloctc\empbisimp \prgzero'$, hence the result holds.

  If $\ctx = \mtctx$, $\tmone = \shift \vark {\tmone'}$ with $\tmzero
  \cloctc\empbisimp \tmone'$, then we have $\reset{\subst \tmzero
    \vark {\lam \vark {\reset {\inctx \ctxzero \varx}}}} \cloctc
  \empbisimp \reset{\subst {\tmone'} \vark {\lam \vark {\reset {\inctx
          \ctxone \varx}}}}$ by Lemma \ref{l:cong-is-subst}, and
  because we have the two reductions $\reset {\inctx \ctxzero {\shift
      \vark \tmzero}} \redcbv \reset{\subst \tmzero \vark {\lam \vark
      {\reset {\inctx \ctxzero \varx}}}}$ and $\reset {\inctx \ctxone
    {\shift \vark \tmone}} \redcbv \reset{\subst {\tmone'} \vark {\lam
      \vark {\reset {\inctx \ctxone \varx}}}}$, we obtain the required
  result.
  
  Suppose $\inctx \ctx {\shift \vark \tmzero} = \app \valzero {\inctx
    {\ctx'}{\shift \vark \tmzero}}$, $\tmone =
  \app{\tmone^1}{\tmone^2}$ with $\valzero \cloctc\empbisimp \tmone^1$
  and $\inctx {\ctx'}{\shift \vark \tmzero} \cloctc\empbisimp
  \tmone^2$. We distinguish two cases. If $\tmone^1 = \valone$, then
  by the induction hypothesis, there exist $\prgzero$, $\prgone$ such
  that $\reset{\inctx \ctxzero {\app \valzero {\inctx {\ctx'}{\shift
          \vark \tmzero}}}} \redcbv \prgzero$, $\reset{\inctx \ctxone
    {\app \valone {\tmone^2}}} \clocbv \prgone$, and $\prgzero \cloctc
  \empbisimp \empbisimp \prgone$, therefore we have the required
  result. Suppose now $\valzero \empbisimp \tmone^1$. By the induction
  hypothesis, there exist $\prgzero$, $\prgzero'$ such that
  $\reset{\inctx \ctxzero {\app \valzero {\inctx {\ctx'}{\shift \vark
          \tmzero}}}} \redcbv \prgzero$, $\reset{\inctx \ctxone {\app
      \valzero {\tmone^2}}} \clocbv \prgzero'$, and $\prgzero \cloctc
  \empbisimp \empbisimp \prgzero'$. From $\valzero \empbisimp
  \tmone^1$, we know that $\reset{\inctx \ctxone {\app \valzero
      {\tmone^2}}} \empbisimp \reset{\inctx \ctxone {\app
      {\tmone^1}{\tmone^2}}}$. By bisimilarity, there exists $\prgone$
  such that $\reset{\inctx \ctxone {\app {\tmone^1}{\tmone^2}}}
  \clocbv \prgone$ and $\prgzero' \empbisimp \prgone$. Therefore we
  have $\prgzero \cloctc \empbisimp \empbisimp \prgone$, hence the
  result holds.

  Suppose $\inctx \ctx {\shift \vark \tmzero} = \app {\inctx
    {\ctx'}{\shift \vark \tmzero}} \tm$, $\tmone =
  \app{\tmone^1}{\tmone^2}$ with $\inctx {\ctx'}{\shift \vark \tmzero}
  \cloctc\empbisimp \tmone^1$ and $\tm \cloctc\empbisimp \tmone^2$. By
  the induction hypothesis, there exist $\prgzero$, $\prgone$ such
  that $\reset{\inctx \ctxzero {\app {\inctx {\ctx'}{\shift \vark
          \tmzero}} \tm}} \redcbv \prgzero$, $\reset{\inctx \ctxone
    {\app {\tmone^1}{\tmone^2}}} \clocbv \prgone$, and $\prgzero
  \cloctc \empbisimp \empbisimp \prgone$, therefore we have the
  required result.
\end{proof}

\begin{lemma}
  \label{l:cong-empbisim-p}
  $\tmzero \empbisimp \tmone$ implies $\inctx \cctx \tmzero \ierel
  \bisimp {\rvals {\cloct \empbisimp}} \inctx \cctx \tmone$.
\end{lemma}

\begin{proof}
  We prove that 
  $$\erel = \{ (\rval {\cloct \empbisimp}, \tmzero, \tmone) \mmid \tmzero
  \cloctc \empbisimp \tmone\} \cup \{ \rval {\cloct \empbisimp} \}$$
  is a bisimulation up to bisimilarity. Note that by definition of
  $\erel$, we have $\tm \ierel \erel {\rval {\cloct \empbisimp}} \tm'$
  iff $\tm \cloctc \empbisimp \tm'$.

  Let $\tmzero \ierel \erel {\rval {\cloct \empbisimp}}
  \tmone$. Suppose $\tmzero$ and $\tmone$ are not both programs. Then
  for all $\ctxzero \clocc\empbisimp \ctxone$, we have $\reset{\inctx
    \ctxzero \tmzero} \ierel \erel {\rval {\cloct \empbisimp}}
  \reset{\inctx \ctxone \tmone}$, hence clause \ref{e:tm-p} is
  satisfied.

  Suppose $\tmzero$ and $\tmone$ are both programs $\prgzero$,
  $\prgone$. By Lemma~\ref{l:decompose-prg}, we have several
  possibilities. If $\prgzero \empbisimp \prgone$, then the result
  holds trivially. Suppose $\prgzero = \reset{\valzero}$, $\prgone =
  \reset {\tmone'}$ with $\valzero \cloctc\empbisimp \tmone'$. By
  Lemma \ref{l:bisim-val-p}, there exists $\valone$ such that
  $\reset{\tmone'} \clocbv \valone$ and $\valzero \cloctc\empbisimp
  \valone$. We have $\{(\valzero, \valone)\} \cup \rval {\cloct
    \empbisimp} = \rval {\cloct \empbisimp} \in \erel$, hence the
  result holds.

  Suppose $\prgzero = \inctx \rctxzero {\reset{\inctx \ctxzero
      {\tmzero'}}}$ and $\prgone = \inctx \rctxone {\reset{\inctx
      \ctxone {\tmone'}}}$, with $\rctxzero \clocc \empbisimp
  \rctxone$, $\ctxzero \clocc\empbisimp \ctxone$, $\tmzero' \empbisimp
  \tmone'$, and $\tmzero \redcbv \tmzero'$ or $\tmzero$ is stuck. By
  Lemma \ref{l:progress-bisim-p}, there exist $\prgzero'$, $\prgone'$
  such that $\reset{\inctx \ctxzero {\tmzero'}} \redcbv \prgzero'$,
  $\reset{\inctx \ctxone {\tmone'}} \clocbv \prgone'$, and $\prgzero'
  \cloctc \empbisimp \empbisimp \prgone'$. By definition of $\cloctc
  \empbisimp$ and Lemma \ref{l:cong-e-ctx-p}, we have $\inctx
  \rctxzero {\prgzero'} \cloctc \empbisimp \empbisimp \inctx \rctxone
            {\prgone'}$. Moreover $\prgzero \redcbv \inctx \rctxzero
            {\prgzero'}$ and $\prgone \clocbv \inctx \rctxone
            {\prgone'}$, hence the result holds.
  
  The last possibility is $\prgzero = \inctx \rctxzero {\reset{\inctx
      \ctxzero \redexzero}}$, $\prgone = \inctx \rctxone
  {\reset{\inctx \ctxone {\tmone'}}}$, with $\rctxzero \clocc
  \empbisimp \rctxone$, $\ctxzero \clocc\empbisimp \ctxone$,
  $\redexzero \cloctc\empbisimp \tmone'$, and $\redexzero
  \not\empbisimp \tmone'$. We discuss the three possible redexes. If
  $\redexzero = \reset\valzero$, then $\tmone' = \reset{\tmone''}$
  with $\valzero \cloctc \empbisimp \tmone''$. By Lemma
  \ref{l:bisim-val-p}, there exists $\valone$ such that
  $\reset{\tmone''} \clocbv \valone$ and $\valzero \cloctc \empbisimp
  \valone$. Then we have $\prgzero \redcbv \inctx \rctxzero {\reset
    {\inctx \ctxzero \valzero}}$ and $\prgone \clocbv \inctx \rctxone
         {\reset {\inctx \ctxone \valone}}$ with $\inctx \rctxzero
         {\reset {\inctx \ctxzero \valzero}} \cloctc\empbisimp \inctx
         \rctxone {\reset {\inctx \ctxone \valone}}$, hence the result
         holds. If $\redexzero = \app {\valzero^1}{\valzero^2}$, then
         $\tmone' = \app{\tmone^1}{\tmone^2}$ with $\valzero^1
         \cloctc\empbisimp \tmone^1$, and $\valzero^2
         \cloctc\empbisimp \tmone^2$. By Lemma
         \ref{l:bisim-beta-redex-p}, there exist $\prgzero'$,
         $\prgone'$ such that $\inctx \ctxzero \redexzero \redcbv
         \prgzero'$, $\inctx \ctxone {\tmone'} \clocbv \prgone'$, and
         $\prgzero' \cloctc\empbisimp \empbisimp \prgone'$. Therefore
         we have $\prgzero \redcbv \inctx \rctxzero {\prgzero'}$,
         $\prgone \clocbv \inctx \rctxone {\prgone'}$, with $\inctx
         \rctxzero {\prgzero'} \cloctc\empbisimp \empbisimp \inctx
         \rctxone {\prgone'}$ (by Lemma \ref{l:cong-e-ctx-p} and
         definition of $\cloctc\empbisimp$), hence the result
         holds. The last case is $\redexzero = \reset{\inctx \ctx
           {\shift \vark {\tmzero'}}}$; then $\tmone' =
         \reset{\tmone''}$ with $\inctx \ctx {\shift \vark {\tmzero'}}
         \cloctc \empbisimp \tmone''$. By Lemma \ref{l:bisim-stuck-p},
         there exist $\prgzero'$, $\prgone'$ such that $\redexzero
         \redcbv \prgzero'$, $\tmone' \clocbv \prgone'$ and $\prgzero'
         \cloctc\empbisimp \empbisimp \prgone'$. Therefore we have
         $\prgzero \redcbv \inctx \rctxzero {\reset{\inctx
             \ctxzero{\prgzero'}}}$, $\prgone \clocbv \inctx \rctxone
                           {\reset{\inctx \ctxone {\prgone'}}}$, with
                           $\inctx \rctxzero {\reset{\inctx \ctxzero
                               {\prgzero'}}} \cloctc\empbisimp
                           \empbisimp \inctx \rctxone {\reset{\inctx
                               \ctxone {\prgone'}}}$ (by Lemma
                           \ref{l:cong-e-ctx-p} and definition of
                           $\cloctc\empbisimp$), hence the result
                           holds.
  
  Finally, let $\lam \varx \tmzero \cloctc \empbisimp \lam \varx
  \tmone$ and $\valzero \cloctc\empbisimp \valone$. By Lemma
  \ref{l:subst-val-p}, we get $\subst \tmzero \varx \valzero \cloctc
  \empbisimp \empbisimp \subst \tmone \varx \valone$, hence the
  required result holds.
\end{proof}

\begin{lemma}
  \label{l:completeness-p}
  We have $\mathord{\ctxequivep} \subseteq
  \mathord{\empbisimp}$. 
\end{lemma}

\begin{proof}
  We prove that $\erel = \{(\rval \ctxequivep, \tmzero, \tmone) \mmid
  \tmzero \ctxequivep \tmone\} \cup \{ \rval \ctxequiv \}$ is a
  big-step environmental bisimulation for programs.

  Let $\tmzero \ierel \erel {\rval \ctxequivep} \tmone$ such that $\tmzero$ and
  $\tmone$ are not both programs. Because $\ctxequivep$ is a congruence
  w.r.t. evaluation contexts, we have $\reset{\inctx \ctxzero \tmzero}
  \ctxequivep \reset{\inctx \ctxone \tmone}$ for all $\ctxzero \clocc{\rval
    \ctxequivep} \ctxone$, i.e., $\reset{\inctx \ctxzero \tmzero} \ierel \erel
  {\rval \ctxequivep} \reset{\inctx \ctxone \tmone}$ as wished.

  Let $\prgzero \ierel \erel {\rval \ctxequivep} \prgone$ such that $\prgzero
  \clocbv \valzero$. We have $\reset{\prgzero} \clocbv \valzero$, so by
  definition of $\ctxequivep$, there exists $\valone$ such that $\reset{\prgone}
  \clocbv \valone$, which implies $\prgone \clocbv \valone$. By
  Lemma~\ref{l:redcbv-in-empbisim-p}, we have $\prgzero \empbisimp \valzero$ and
  $\prgone \empbisimp \valone$, which implies $\prgzero \ctxequivep \valzero$
  and $\prgone \ctxequivep \valone$ by soundness of $\empbisimp$. Transitivity
  of $\ctxequivep$ gives $\valzero \ctxequivep \valone$, hence we have $\rval
  \ctxequivep \cup \{(\valzero, \valone)\} = \mathord{\rval \ctxequivep} \in
  \erel$, as wished.

  Let $\lam \varx \tmzero \ctxequivep \lam \varx \tmone$ and $\valzero
  \cloctc{\rval \ctxequivep} \valone$. By congruence of $\ctxequivep$, we
  have $\valzero \ctxequivep \valone$, and also $\app {\lamp \varx
    \tmzero} \valzero \ctxequivep \app {\lamp \varx \tmone}
  \valone$. Because $\app {\lamp \varx \tmzero} \redcbv \subst \tmzero
  \varx \valzero$, $\app {\lamp \varx \tmone} \valone \redcbv \subst
  \tmone \varx \valone$ and $\redcbv \subseteq \mathord{\empbisimp}
  \subseteq \mathord{\ctxequivep}$, we have $\subst \tmzero \varx
  \valzero \ctxequivep \subst \tmone \varx \valone$, i.e., $\subst
  \tmzero \varx \valzero \ierel \erel {\rval \ctxequivep} \subst \tmone
  \varx \valone$, as wished.

\end{proof} 

\section{Bisimulation proofs}

\subsection{Proof of the $\shift \vark {\app \vark \tm} \eqKH \tm$ Axiom}
\label{a:skkt-open}

Let
\begin{multline*}
 \env_1 = \{ (\lam {\vect \varx}{\tm \subs^0_1 \ldots \subs^0_n}, \lam
{\vect \varx}{\shift \vark {\app \vark {\tm \subs^1_1 \ldots \subs^1_n}}}) \mmid \\
 \fv \tm \subseteq \vect \varx \cup \{\varx_1 \ldots \varx_n \}, \vark \notin \fv
\tm \}
\end{multline*}
and
$$\env_2 = \{(\val \subs^0_1 \ldots \subs^0_m, \val \subs^1_1 \ldots
\val \subs^1_m) \mmid \fv \val \subseteq \{\varx_1 \ldots \varx_m \} \},$$
where $\subs^0_i$ and $\subs^1_i$ are of the form $\subst \cdot {\varx_i}{\lam
  {\vect y}{\tm_i}}$ and $\subst \cdot {\varx_i}{\lam {\vect y}{\shift \vark
    {\app \vark {\tm_i}}}}$ respectively for some $\tm_i$ such that $\vark
\notin \fv{\tm_i}$. Let $\env = \env_1 \cup \env_2$. We prove that
\begin{multline*}
\erel = \{ ( \env, \val_1, \val_2) \mmid \val_1 \mathrel{\env_1} \val_2 \} \\
\cup \{ (\env, \tm \subs^0_1 \ldots \subs^0_m, \shift \vark {\app \vark \tm}
\subs^1_1 \ldots \subs^1_m \mmid \fv \tm \subseteq \{\varx_1 \ldots \varx_n\},
\vark \notin \fv \tm \} \\
\cup \{ ( \env, \tm \subs^0_1 \ldots \subs^0_n, \tm \subs^1_1 \ldots \subs^1_n),
  \mmid \fv \tm \subseteq \{\varx_1 \ldots \varx_n\} \}
\end{multline*}
is a bisimulation for programs up to bisimilarity. 

Let $\val_1 \mathrel{\env_1} \val_2$, and let $\ctx_1 \clocc\env \ctx_2$. Then
we have 
$\val_1 = \lam {\vect \varx}{\tm \subs^0_1 \ldots \subs^0_n}$, $\val_2 = \lam
{\vect \varx}{\shift \vark {\app \vark {\tm \subs^1_1 \ldots \subs^1_n}}}$, and
$\ctx_1$ and $\ctx_2$ can be rewritten $\ctx_1 = \ctx \subs^0_{n+1} \ldots
\subs^0_m$, and $\ctx_2 = \ctx \subs^1_{n+1} \ldots \subs^1_m$ for some
$\ctx$. We can rewrite $\reset{\inctx {\ctx_1}{\val_1}}$ and $\reset{\inctx
  {\ctx_2}{\val_2}}$ into
\begin{align*}
\reset{\inctx {\ctx_1}{\val_1}} &=
\subst{\reset{\inctx \ctx y}} y {\lam {\vect \varx} \tm}\subs^0_1 \ldots
\subs^0_m \\
\reset{\inctx {\ctx_2}{\val_2}} &= \subst{\reset{\inctx \ctx y}}
y {\lam {\vect \varx}{\shift \vark {\app \vark \tm}}}\subs^1_1 \ldots \subs^1_m
\end{align*}
for some fresh $y$. Hence, we have $\reset {\inctx {\ctx_1}{\val_1}} \ierel
\erel \env \reset {\inctx {\ctx_2}{\val_2}}$, as wished.

Let $\tmzero = \tm \subs^0_1 \ldots \subs^0_n \ierel \erel \env \shift \vark
{\app \vark \tm} \subs^1_1 \ldots \subs^1_n = \tmone$, and let $\ctx_1
\clocc\env \ctx_2$. Then $\ctx_1$ and $\ctx_2$ can be rewritten $\ctx_1 = \ctx
\subs^0_{n+1} \ldots \subs^0_m$, and $\ctx_2 = \ctx \subs^1_{n+1} \ldots
\subs^1_m$ for some $\ctx$. Therefore
\begin{align*}
\reset{\inctx {\ctx_1}{\tmzero}} &=
\reset{\inctx \ctx \tm}\subs^0_1 \ldots
\subs^0_m \\
\reset{\inctx {\ctx_2}{\tmone}} &= \reset{\inctx \ctx {\shift \vark {\app \vark
      \tm}}} \subs^1_1 \ldots \subs^1_m
\end{align*}
but we have $\reset{\inctx \ctx {\shift \vark {\app \vark \tm}}} \subs^1_1
\ldots \subs^1_m \empbisimp \reset{\inctx \ctx \tm} \subs^1_1 \ldots \subs^1_m$
by the same reasoning as in Section~\ref{s:examples-p}, therefore we have
$\reset{\inctx {\ctx_1}{\tmzero}} \ierel \erel \env \empbisimp \reset{\inctx
  {\ctx_2}{\tmone}}$ as required. 

Let $\tmzero = \tm \subs^0_1 \ldots \subs^0_n \ierel \erel \env \tm
\subs^1_1 \ldots \subs^1_n = \tmone$. If $\tm$ is not a program, then
for all $\ctx_1 \clocc\env \ctx_2$, we can show that $\inctx {\ctx_1}
\tmzero \ierel \erel \env \inctx {\ctx_2} \tmone$ by rewriting
$\ctx_1$ and $\ctx_2$ as in the previous cases. If $\tm$ is a program
$\prg$, then we distinguish several cases. If $\prg \redcbv \prg'$,
then we conclude easily. If $\prg = \reset\val$, then $\tmzero$ and
$\tmone$ reduce to values related by $\env_2$. If $\prg =
\reset{\varx_j}$, then $\tmzero$ and $\tmone$ reduce to values related
by $\env_1$. If $\prg = \reset{\inctx \rctx {\app {\varx_j} \val}}$,
then $\tmzero = \reset{\inctx \rctx {\app {\lamp {\vect y}{\tm_j}}
    \val}}\subs^0_1 \ldots \subs^0_n$ and $\tmone = \reset{\inctx
  \rctx {\app {\lamp {\vect y}{\shift \vark {\app \vark {\tm_j}}}}
    \val}}\subs^1_1 \ldots \subs^1_n$. If $\vect y = y_0 \cup
\vect{y'}$ with $\vect{y'}$ not empty, then $\tmzero \redcbv
\reset{\inctx \rctx {\lam {\vect {y'}}{\subst {\tm_j}{y_0} \val}}
}\subs^0_1 \ldots \subs^0_n$ and $\tmone \redcbv \reset{\inctx \rctx
  {\lam {\vect {y'}}{\shift \vark {\app \vark {\subst {\tm_j}{y_0}
          \val}}}} }\subs^1_1 \ldots \subs^1_n$. The two resulting
terms can be rewritten into $\prg' \subs^0_1 \ldots \subs^0_j
\subs^0_{j'} \ldots \subs^0_n$ and $\prg' \subs^1_1 \ldots \subs^1_j
\subs^1_{j'} \ldots \subs^1_n$, where we have $\prg' = \reset {\inctx
  \rctx {\varx_{j'}}}$, $\subs^0_{j'} = \subst \cdot
     {\varx_{j'}}{\lam{\vect{y'}}{\subst {\tm_j}{y_0} \val}}$,
     $\subs^1_{j'} = \subst \cdot {\varx_{j'}}{\lam {\vect
         {y'}}{\shift \vark {\app \vark {\subst {\tm_j}{y_0}
             \val}}}}$, and $\varx_{j'}$ is fresh. If $\vect y = y$,
     then $\tmzero \redcbv \reset{\inctx \rctx {\subst {\tm_j}{y_0}
         \val}}\subs^0_1 \ldots \subs^0_n = \prgzero$ and $\tmone
     \redcbv \reset{\inctx \rctx {\shift \vark {\app \vark {\subst
             {\tm_j}{y_0} \val}}}}\subs^1_1 \ldots \subs^1_n =
     \prgone$. By the same reasoning as in Section~\ref{s:examples-p}
     (and with some case analysis on $\rctx$), we have $\reset{\inctx
       \rctx {\shift \vark {\app \vark {\subst {\tm_j}{y_0}
             \val}}}}\subs^1_1 \ldots \subs^1_n
     \empbisimp\reset{\inctx \rctx {\subst {\tm_j}{y_0}
         \val}}\subs^1_1 \ldots \subs^1_n$, therefore we have
     $\prgzero \ierel \erel \env \empbisimp \prgone$, as wished.

Let $\lam {\vect \varx}{\tm \subs^0_1 \ldots \subs^0_n} \mathrel{\env_1} \lam
{\vect \varx}{\shift \vark {\app \vark {\tm \subs^1_1 \ldots \subs^1_n}}}$ and
$\valzero \clocc\env \valone$. Then $\valzero = \val \subs^0_{n+1} \ldots
\subs^0_m$ and $\valone = \val \subs^1_{n+1} \ldots \subs^1_m$ for some
$\val$. If $\vect \varx = y \cup \vect{\varx'}$ with $\vect{\varx'}$ not empty,
then we have $\lam {\vect{\varx'}}{\subst \tm y \val \subs^0_1 \ldots \subs^0_m}
\ierel \erel \env \lam {\vect {\varx'}}{\shift \vark {\app \vark {\subst \tm y
      \val \subs^1_1 \ldots \subs^1_m}}}$. If $\vect \varx = y$, then we also
have $\subst \tm y \val \subs^0_1 \ldots \subs^0_m \ierel \erel \env \shift
\vark {\app \vark {\subst \tm y \val \subs^1_1 \ldots \subs^1_m}}$, therefore
the result holds in both cases.

Let $\lam \varx \tm \subs^0_1 \ldots \subs^0_n \mathrel{\env_2} \lam \varx \tm \subs^1_1 \ldots \val
\subs^1_n$ and $\valzero \clocc\env \valone$. Then $\valzero = \val \subs^0_{n+1} \ldots
\subs^0_m$ and $\valone = \val \subs^1_{n+1} \ldots \subs^1_m$ for some
$\val$. Then $\subst \tm \varx \val \subs^0_1 \ldots \subs^0_m \ierel \erel \env
\subst \tm \varx \val \subs^1_1 \ldots \subs^1_m$ holds.

\subsection{Proof of the S-tail Axiom}
\label{a:s-tail}

Let $\tmzero$ and $\tmone$ such that $\vark \notin \fv\tmone$. We want
to show that $\app{\lamp \varx{\shift \vark \tmzero}} \tmone
\empbisimp \shift\vark {\appp {\lamp \varx \tmzero} \tmone}$. To this
end, we need to plug both terms in some context $\reset{\ctx}$, and
compare $\reset{\inctx \ctx {\app{\lamp \varx{\shift \vark \tmzero}}
    \tmone}}$ with $\reset{\inctx \ctx {\shift\vark {\appp {\lamp
        \varx \tmzero} \tmone}}}$. The second term reduces to
$\reset{\app {\lamp \varx {\subst \tmzero \vark {\lam y {\reset
          {\inctx \ctx y}}}}} \tmone}$, so we in fact prove the
following result.

\begin{lemma}
  We have $\reset{\inctx \ctx {\app{\lamp \varx{\shift \vark \tmzero}} \tmone}}
  \empbisim \reset{\app {\lamp \varx {\subst \tmzero \vark {\lam y {\reset
            {\inctx \ctx y}}}}} \tmone}$.
\end{lemma}

\begin{proof}
  To make the proof easier to follow, we introduce some notations. We write
  $\vect \cdot$ for a sequence of entities (e.g., $\vect \ctx$ for a sequence of
  contexts). We write $\bm\ctx$ for $\inctx \ctx {\app{\lamp \varx{\shift \vark
        \tmzero}} \mtctx}$ and $\bm{\ctx'}$ for $\app {\lamp \varx {\subst
      \tmzero \vark {\lam y {\reset {\inctx \ctx y}}}}} \mtctx$, so the problem
  becomes relating $\reset {\inctx {\bm \ctx} \tmone}$ and $\reset {\inctx {\bm
      {\ctx'}} \tmone}$.
   
  Next, given a sequence $\ctx_0 \ldots \ctx_i$ of contexts such that
  $\fv{\ctx_j} \subseteq \{\vark_0 \ldots \vark_{j-1} \}$ for all $0 \leq j \leq
  i$, we inductively define families of substitutions $\subs^{\vect \ctx}_0
  \ldots \subs^{\vect \ctx}_i$, $\delta^{\vect \ctx}_0 \ldots \delta^{\vect
    \ctx}_i$ as follows:

  $
  \begin{array}{rcl}
    \subs^{\vect \ctx}_0 & = & \subst \cdot {\vark_0}{\lam y {\reset {\inctx {\bm
          \ctx}{\inctx \ctxzero y}}}}\\
    \delta^{\vect \ctx}_0 & = & \subst \cdot {\vark_0}{\lam y {\reset {\inctx {\bm
          {\ctx'}}{\inctx \ctxzero y}}}}\\
    \subs^{\vect \ctx}_j & = & \subst \cdot {\vark_j}{\lam y {\reset {\inctx {\bm
          \ctx}{\inctx {\ctx_j\subs^{\vect \ctx}_0 \ldots \subs^{\vect
    \ctx}_{j-1}} y}}}} \mbox{ if } j > 0 \\
    \delta^{\vect \ctx}_j & = & \subst \cdot {\vark_j}{\lam y {\reset {\inctx {\bm
          {\ctx'}}{\inctx {\ctx_j\delta^{\vect \ctx}_0 \ldots \delta^{\vect
    \ctx}_{j-1}} y}}}} \mbox{ if } j > 0
  \end{array}
  $

  Finally, given a term $\tm$ and a sequence of contexts $\rctx_0 \ldots
  \rctx_i$, we inductively define families of terms $\tms^{\tm, \vect \rctx}_0
  \ldots \tms^{\tm, \vect \rctx}_i$, $\tmu^{\tm, \vect \rctx}_0 \ldots
  \tmu^{\tm, \vect \rctx}_i$ as follows:

  $
  \begin{array}{rclrcl}
    \tms^{\tm, \vect \rctx}_0 & = & \inctx {\rctx_0}{\reset {\inctx{\bm \ctx}
        \tm}} 
    & \qquad \tms^{\tm, \vect \rctx}_j & = & \inctx {\rctx_j}{\reset {\inctx{\bm \ctx}
        {\tms^{\tm, \vect \rctx}_{j-1}}}} \mbox{ if } j > 0 \\
    \tmu^{\tm, \vect \rctx}_0 & = & \inctx {\rctx_0}{\reset {\inctx{\bm {\ctx'}}
        \tm}} &
    \qquad \tmu^{\tm, \vect \rctx}_0 & = & \inctx {\rctx_j}{\reset {\inctx{\bm {\ctx'}}
        {\tmu^{\tm, \vect \rctx}_{j-1}}}} \mbox{ if } j > 0
  \end{array}
  $

  Note that the term we want to relate are $\tms^{\tmone, \mtctx}_0$ and
  $\tmu^{\tmone, \mtctx}_0$. We let $\env$ ranges over environments of the form
  $\{ (\val \subs^{\vect \ctx}_0 \ldots \subs^{\vect \ctx}_i, \val \delta^{\vect
    \ctx}_0 \ldots \delta^{\vect \ctx}_i) \mmid \fv \val \subseteq \{ \vark_0 \ldots
  \vark_i \} \} \cup \{ (\vark_j\subs^{\vect \ctx}_j, \vark_j\delta^{\vect
    \ctx}_j) \}$. We prove that the relation 
  \begin{multline*}
    \erel = \{(\env, \tm \subs^{\vect \ctx}_0 \ldots \subs^{\vect
      \ctx}_i, \tm \delta^{\vect \ctx}_0 \ldots \delta^{\vect \ctx}_i) \mmid
    \fv \tm \subseteq \{ \vark_0 \ldots \vark_i \} \} \cup \\ 
    \{(\env, \reset {\tms^{\tm, \vect \rctx}_i}\subs^{\vect {\ctx'}}_0 \ldots \subs^{\vect
      {\ctx'}}_j, \reset {\tmu^{\tm, \vect \rctx}_i}\delta^{\vect {\ctx'}}_0
    \ldots \delta^{\vect {\ctx'}}_j) \mmid
    \fv \tm \cup \fv{\vect \rctx} \subseteq \{ \vark_0 \ldots \vark_j \} \} \cup
    \{\env \}
  \end{multline*}
  is a bisimulation for programs. Let $\tm \subs^{\vect \ctx}_0 \ldots
  \subs^{\vect \ctx}_i \ierel \erel \env \tm \delta^{\vect \ctx}_0 \ldots
  \delta^{\vect \ctx}_i$ where $\tm$ is not a program. Let $\ctxzero \clocc \env
  \ctxone$; by definition of $\env$, we have $\ctxzero = \ctx' \subs^{\vect
    {\ctx'}}_0 \ldots \subs^{\vect {\ctx''}}_j$ and $\ctxone = \ctx'
  \delta^{\vect \ctx}_0 \ldots \delta^{\vect {\ctx''}}_j$ for some $\ctx'$,
  $\vect{\ctx''}$. With some renumbering and rewriting, we have
  $\reset{\inctx \ctxzero {\tm \subs^{\vect \ctx}_0 \ldots \subs^{\vect
        \ctx}_i}} = \reset{\inctx {\ctx'} \tm}\subs^{\vect \ctx,\vect
    {\ctx''}}_0 \ldots \subs^{\vect \ctx,\vect {\ctx''}}_{i+j+1}$ and
  $\reset{\inctx \ctxone {\tm \subs^{\vect \ctx}_0 \ldots \subs^{\vect \ctx}_i}}
  = \reset{\inctx {\ctx'} \tm}\delta^{\vect \ctx,\vect {\ctx''}}_0 \ldots
  \delta^{\vect \ctx,\vect {\ctx''}}_{i+j+1}$: the
  two terms are in $\erel$, as wished.\\

  Let $\reset \tm \subs^{\vect \ctx}_0 \ldots \subs^{\vect \ctx}_i \ierel \erel
  \env \reset \tm \delta^{\vect \ctx}_0 \ldots \delta^{\vect \ctx}_i$. We have
  three cases for $\tm$.

  If $\reset{\tm} \redcbv \reset{\tm'}$, we still have $\reset {\tm'}
  \subs^{\vect \ctx}_0 \ldots \subs^{\vect \ctx}_i \ierel \erel \env \reset
  {\tm'} \delta^{\vect \ctx}_0 \ldots \delta^{\vect \ctx}_i$. If $\reset \tm
  \redcbv \val$ or $\tm = \vark_j$, then both terms reduce to values that are in
  $\env$, by definition of $\env$.

  If $\tm = \inctx \rctx {\app {\vark_j} \val}$, then 
  
  $
  \begin{array}{rcl}
    \reset \tm \subs^{\vect \ctx}_0 \ldots
    \subs^{\vect \ctx}_i &=& \inctx {\rctx\subs^{\vect \ctx}_0 \ldots \subs^{\vect
        \ctx}_i}{\app {\lam y {\reset {\inctx {\bm
              \ctx}{\inctx {\ctx_j\subs^{\vect \ctx}_0 \ldots \subs^{\vect
                  \ctx}_{j-1}} y}}}}{\val \subs^{\vect \ctx}_0 \ldots
        \subs^{\vect \ctx}_i}} \\
    \reset \tm \delta^{\vect \ctx}_0 \ldots
    \delta^{\vect \ctx}_i &=& \inctx {\rctx\delta^{\vect \ctx}_0 \ldots \delta^{\vect
        \ctx}_i}{\app {\lam y {\reset {\inctx {\bm
              {\ctx'}}{\inctx {\ctx_j\delta^{\vect \ctx}_0 \ldots \delta^{\vect
                  \ctx}_{j-1}} y}}}}{\val \delta^{\vect \ctx}_0 \ldots
        \delta^{\vect \ctx}_i}}    
  \end{array}
  $

 \noindent  Reducing the $\beta$-redex in both terms, we obtain

  $
  \begin{array}{rcl}
    \reset \tm \subs^{\vect \ctx}_0 \ldots
    \subs^{\vect \ctx}_i &\redcbv& \inctx {\rctx\subs^{\vect \ctx}_0 \ldots \subs^{\vect
        \ctx}_i}{\reset {\inctx {\bm
              \ctx}{\inctx {\ctx_j\subs^{\vect \ctx}_0 \ldots \subs^{\vect
                  \ctx}_{j-1}} {\val \subs^{\vect \ctx}_0 \ldots
        \subs^{\vect \ctx}_i}}}} \\
    \reset \tm \delta^{\vect \ctx}_0 \ldots
    \delta^{\vect \ctx}_i &\redcbv& \inctx {\rctx\delta^{\vect \ctx}_0 \ldots \delta^{\vect
        \ctx}_i}{\reset {\inctx {\bm
              {\ctx'}}{\inctx {\ctx_j\delta^{\vect \ctx}_0 \ldots \delta^{\vect
                  \ctx}_{j-1}} {\val \delta^{\vect \ctx}_0 \ldots
        \delta^{\vect \ctx}_i}}}}
  \end{array}
  $

  \noindent The resulting terms can be written $\reset {\tms^{\tm', \rctx}_0}\subs^{\vect
    {\ctx}}_0 \ldots \subs^{\vect {\ctx}}_j$ and $\reset {\tmu^{\tm', \vect
      \rctx}_0}\delta^{\vect {\ctx}}_0 \ldots \delta^{\vect {\ctx}}_j$, with
  $\tm' = \inctx {\ctx_j}{\val}$, therefore we obtain terms in $\ierel \erel
  \env$.\\

  Let $\reset {\tms^{\tm, \vect \rctx}_i}\subs^{\vect {\ctx'}}_0 \ldots
  \subs^{\vect {\ctx'}}_j \ierel \erel \env \reset {\tmu^{\tm, \vect
      \rctx}_i}\delta^{\vect {\ctx'}}_0 \ldots \delta^{\vect {\ctx'}}_j$. One
  can check that the reductions from terms of the form $\tms^{\tm, \vect
    \rctx}_i$, $\tmu^{\tm, \vect \rctx}_i$ come from respectively $\tms^{\tm
    \vect \rctx}_0$ and $\tmu^{\tm, \vect \rctx}_0$, and the transitions from
  these two terms come from $\tm$. We have several cases for~$\tm$.
  If $\tm \redcbv \tm'$, then we still have $\reset {\tms^{\tm', \vect
      \rctx}_i}\subs^{\vect {\ctx'}}_0 \ldots \subs^{\vect {\ctx'}}_j \ierel
  \erel \env \reset {\tmu^{\tm', \vect \rctx}_i}\delta^{\vect {\ctx'}}_0 \ldots
  \delta^{\vect {\ctx'}}_j$.

  If $\tm = \val$, then $\reset {\tms^{\val, \vect
      \rctx}_0}\subs^{\vect {\ctx'}}_0 \ldots \subs^{\vect {\ctx'}}_j
  = \reset{\inctx \rctxzero {\reset {\inctx {\bm\ctx}
        \val}}}\subs^{\vect {\ctx'}}_0 \ldots \subs^{\vect {\ctx'}}_j$
  and we also have $\reset {\tmu^{\tm', \vect \rctx}_0}\delta^{\vect
    {\ctx'}}_0 \ldots \delta^{\vect {\ctx'}}_j = \reset{\inctx
    \rctxzero {\reset {\inctx {\bm{\ctx'}} \val}}}\delta^{\vect
    {\ctx'}}_0 \ldots \delta^{\vect {\ctx'}}_j$. It is easy to check
  that $\reset{\inctx {\bm\ctx} \val}$ and $\reset{\inctx {\bm{\ctx'}}
    \val}$ reduce to the same term $\reset{\subst{\subst \tmzero \vark
      {\lam y {\reset{\inctx \ctx y}}}} \varx \val}$, written
  $\tm'$. Then we have $\reset {\tms^{\val, \vect
      \rctx}_0}\subs^{\vect {\ctx'}}_0 \ldots \subs^{\vect {\ctx'}}_j
  \redcbv \reset{\inctx \rctxzero {\tm'}}\subs^{\vect {\ctx'}}_0
  \ldots \subs^{\vect {\ctx'}}_j$, and also $\reset {\tmu^{\tm', \vect
      \rctx}_0}\delta^{\vect {\ctx'}}_0 \ldots \delta^{\vect
    {\ctx'}}_j \redcbv \reset{\inctx \rctxzero {\tm'}}\delta^{\vect
    {\ctx'}}_0 \ldots \delta^{\vect {\ctx'}}_j$; the two resulting
  terms are in the first set of $\erel$. If $i >0$, one can check that
  $\reset {\tms^{\val, \vect \rctx}_i}\subs^{\vect {\ctx'}}_0 \ldots
  \subs^{\vect {\ctx'}}_j \redcbv \reset {\tms^{\reset{\inctx
        \rctxzero {\tm'}}, \vect {\rctx'}}_{i-1}}\subs^{\vect
    {\ctx'}}_0 \ldots \subs^{\vect {\ctx'}}_j$ and we also have
  $\reset {\tmu^{\val, \vect \rctx}_i}\delta^{\vect {\ctx'}}_0 \ldots
  \delta^{\vect {\ctx'}}_j \redcbv \reset {\tmu^{\reset{\inctx
        \rctxzero {\tm'}}, \vect {\rctx'}}_{i-1}}\delta^{\vect
    {\ctx'}}_0 \ldots \delta^{\vect {\ctx'}}_j$, where $\vect{\rctx'}
  = \rctxone \ldots \rctx_i$ (the first context $\rctxzero$ is removed
  from the sequence). We obtain terms that are in the second set of
  $\erel$. In both cases, the resulting terms are in $\erel$. The
  reasoning is the same if $\tm = \vark_l$ for some $0 \leq l \leq j$.

  If $\tm = \inctx{\ctx'_{j+1}}{\shift {\vark_{j+1}}{\tm'}}$, then 
  \begin{multline*}
    \reset {\tms^{\tm, \vect \rctx}_0}\subs^{\vect {\ctx'}}_0 \ldots
    \subs^{\vect {\ctx'}}_j = \reset{\inctx \rctxzero {\reset {\inctx
          {\bm\ctx}{\inctx{\ctx'_{j+1}}{\shift
              {\vark_{j+1}}{\tm'}}}}}}\subs^{\vect {\ctx'}}_0 \ldots
    \subs^{\vect {\ctx'}}_j \\ \redcbv \reset{\inctx \rctxzero {\reset
        {\tm'}}}\subs^{\vect {\ctx'}}_0 \ldots \subs^{\vect {\ctx'}}_j
    \subs^{\vect {\ctx'}, \ctx'_{j+1}}_{j+1} 
  \end{multline*}
  and 
  \begin{multline*}
    \reset {\tmu^{\tm, \vect \rctx}_0}\delta^{\vect {\ctx'}}_0 \ldots
    \delta^{\vect {\ctx'}}_j = \reset{\inctx \rctxzero {\reset {\inctx
          {\bm{\ctx'}}{\inctx{\ctx'_{j+1}}{\shift
              {\vark_{j+1}}{\tm'}}}}}}\delta^{\vect {\ctx'}}_0 \ldots
    \delta^{\vect {\ctx'}}_j \\ \redcbv \reset{\inctx \rctxzero {\reset
        {\tm'}}}\delta^{\vect {\ctx'}}_0 \ldots \delta^{\vect {\ctx'}}_j
    \delta^{\vect {\ctx'}, \ctx'_{j+1}}_{j+1}
  \end{multline*}
  therefore $\reset {\tms^{\val, \vect \rctx}_0}\subs^{\vect
    {\ctx'}}_0 \ldots \subs^{\vect {\ctx'}}_j$ and $\reset
  {\tmu^{\val, \vect \rctx}_0}\delta^{\vect {\ctx'}}_0 \ldots
  \delta^{\vect {\ctx'}}_j$ reduce to terms of the form
  $\reset{\tm''}\subs^{\vect {\ctx'}}_0 \ldots \subs^{\vect
    {\ctx'}}_{j+1}$ and $\reset{\tm''}\delta^{\vect {\ctx'}}_0 \ldots
  \delta^{\vect {\ctx'}}_{j+1}$, that are in $\ierel \erel \env$. If
  $i>0$, then one can check that $\reset {\tms^{\val, \vect
      \rctx}_i}\subs^{\vect {\ctx'}}_0 \ldots \subs^{\vect {\ctx'}}_j
  \redcbv\reset {\tms^{\reset{\inctx \rctxzero {\tm'}}, \vect
      {\rctx'}}_{i-1}}\subs^{\vect {\ctx'}}_0 \ldots \subs^{\vect
    {\ctx'}}_{j+1}$ and also $\reset {\tmu^{\val, \vect
      \rctx}_i}\delta^{\vect {\ctx'}}_0 \ldots \delta^{\vect
    {\ctx'}}_j \redcbv \reset {\tmu^{\reset{\inctx \rctxzero {\tm'}},
      \vect {\rctx'}}_{i-1}}\delta^{\vect {\ctx'}}_0 \ldots
  \delta^{\vect {\ctx'}}_{j+1}$, where $\vect{\rctx'} = \rctxone
  \ldots \rctx_i$, so the resulting terms are in $\ierel \erel \env$.

  If $\tm = \inctx {\rctx_{i+1}}{\app {\vark_l} \val}$ (with $1 \leq l \leq j$),
  then 
  \begin{multline*}
    \reset {\tms^{\tm, \vect \rctx}_0}\subs^{\vect {\ctx'}}_0 \ldots
    \subs^{\vect {\ctx'}}_j = \reset{\inctx \rctxzero {\reset {\inctx
          {\bm\ctx}{\inctx{\rctx_{i+1}}{\app{\lamp y {\reset {\inctx {\bm
          \ctx}{\inctx {\ctx_l} y}}}} \val}}}}}\subs^{\vect {\ctx'}}_0 \ldots
    \subs^{\vect {\ctx'}}_j \\ \redcbv \reset{\inctx \rctxzero {\reset {\inctx
          {\bm\ctx}{\inctx{\rctx_{i+1}}{\reset {\inctx {\bm
          \ctx}{\inctx {\ctx_l} \val}}}} }}}\subs^{\vect {\ctx'}}_0 \ldots
    \subs^{\vect {\ctx'}}_j = \reset {\tms^{\inctx {\ctx_l} \val, \vect {\rctx'}}_1}\subs^{\vect {\ctx'}}_0 \ldots
    \subs^{\vect {\ctx'}}_j 
  \end{multline*}
  and 
  \begin{multline*}
    \reset {\tmu^{\tm, \vect \rctx}_0}\delta^{\vect {\ctx'}}_0 \ldots
    \delta^{\vect {\ctx'}}_j = \reset{\inctx \rctxzero {\reset {\inctx
          {\bm{\ctx'}}{\inctx{\rctx_{i+1}}{\app{\lamp y {\reset {\inctx {\bm
          {\ctx'}}{\inctx {\ctx_l} y}}}} \val}}}}}\delta^{\vect {\ctx'}}_0 \ldots
    \delta^{\vect {\ctx'}}_j \\ \redcbv \reset{\inctx \rctxzero {\reset {\inctx
          {\bm{\ctx'}}{\inctx{\rctx_{i+1}}{\reset {\inctx {\bm
          {\ctx'}}{\inctx {\ctx_l} \val}}}} }}}\delta^{\vect {\ctx'}}_0 \ldots
    \delta^{\vect {\ctx'}}_j = \reset {\tmu^{\inctx {\ctx_l} \val, \vect {\rctx'}}_1}\delta^{\vect {\ctx'}}_0 \ldots
    \delta^{\vect {\ctx'}}_j 
  \end{multline*}
  with $\vect {\rctx'} = \rctx_{i+1},\rctx_0, \ldots \rctx_i$, so the resulting
  terms are in $\ierel \erel \env$. If $i>0$, then $\reset {\tms^{\tm, \vect
      \rctx}_0}\subs^{\vect {\ctx'}}_0 \ldots \subs^{\vect {\ctx'}}_j \redcbv
  \reset {\tms^{\inctx {\ctx_l} \val, \vect {\rctx'}}_{i+1}}\subs^{\vect
    {\ctx'}}_0 \ldots \subs^{\vect {\ctx'}}_j$, and we have also $\reset {\tmu^{\tm, \vect
      \rctx}_0}\delta^{\vect {\ctx'}}_i \ldots \delta^{\vect {\ctx'}}_j \redcbv
  \reset {\tmu^{\inctx {\ctx_l} \val, \vect {\rctx'}}_{i+1}}\delta^{\vect
    {\ctx'}}_0 \ldots \delta^{\vect {\ctx'}}_j$, so the resulting terms are in
  $\ierel \erel \env$, as required. \\

  Finally, let $\lam \varx \tmzero \mathrel\env \lam \varx \tmone$ and $\valzero
  \cloctc \env \valone$. It is easy to check that by definition of $\env$, the
  two terms $\subst \tmzero \varx \valzero$ and $\subst \tmone \varx \valone$
  are of the form $\tm' \subs^{\vect {\ctx'}}_0 \ldots \subs^{\vect \ctx}_i$ and
  $\tm' \delta^{\vect {\ctx'}}_0 \ldots \delta^{\vect \ctx}_i$. 
 
\end{proof}

\end{document}